\documentclass{autart}     
\usepackage{float}
\usepackage{natbib}
\usepackage{graphicx} 
    \usepackage{titlesec}
\usepackage{amsmath}
\usepackage{enumerate}
\usepackage{tikz}
\usepackage{algorithmic}
\usepackage{amssymb,amsfonts,comment}
\usepackage{algorithm}
\usepackage{tikz}
\usepackage{subcaption}
\usetikzlibrary{arrows.meta,positioning,fit,backgrounds}
\usepackage{cases}
\DeclareMathOperator*{\argmin}{arg\,min}
\definecolor{myred}{RGB}{200,30,30}
\definecolor{myyellow}{RGB}{200,200,30}
\definecolor{mygreen}{RGB}{30,200,30}

\newcommand{\sz}[1]{{\color{red}{#1}}}

\usepackage{hyperref}
\hypersetup{
	colorlinks = true,
	citecolor = cyan,
}
\usepackage{wrapfig}

\newtheorem{remark}{Remark}

\newtheorem{assumption}{Assumption}
\newtheorem{example}{Example}
\newtheorem{definition}{Definition}
\newtheorem{lemma}{Lemma}
\newtheorem{corollaryx}{Corollary}
\newtheorem{theorem}{Theorem}

\newtheorem{problem}{Problem}

\makeatletter
\renewenvironment{pf}
{\par\noindent\textbf{Proof.}\ }
  {\hfill$\Box$\par}
\makeatother

\begin{document}
	\begin{frontmatter}
		
		\title{Stochastic Adaptive Control for   Systems with  Nonlinear Parameterization: Almost Sure Stability  and  Tracking 
        } 
		\thanks[footnoteinfo]
        {This work is supported by The Wallenberg AI, Autonomous Systems
and Software Program (WASP) which is funded by the Knut and Alice
Wallenberg Foundation.}
		\author[KTH]{Lantian Zhang
        }\ead{lantian@kth.se},    
		\author[EECS]{Bo Wahlberg}\ead{bo@kth.se},              
		\author[KTH]{Silun Zhang
        }\ead{silunz@kth.se}  
		
		\address[KTH]{Department of Mathematics,  KTH Royal Institute of Technology, Stockholm,  Sweden}
		
		\address[EECS]{Division of Decision and Control Systems, KTH Royal Institute of Technology, Stockholm, Sweden}
		
		\begin{keyword}                           
             Stochastic nonlinear systems, Adaptive control,  
System identification, Adaptive tracking.
		\end{keyword}                             

		\begin{abstract} 
    This paper concerns the adaptive control problem for a class of nonlinear stochastic systems in which the state update is given by a nonlinear function of linear dynamics plus additive stochastic noise. Such systems arise in a wide range of applications, including recurrent neural networks, social dynamics, and signal processing.
    Despite their importance, adaptive control for these systems remains relatively unexplored in the literature. This gap is primarily due to the inherently nonconvex dependence of the system dynamics on unknown parameters, which significantly complicates both controller design and analysis. To address these challenges, we propose an online nonlinear weighted least-squares (WLS)–based parameter estimation algorithm and establish the global strong consistency of the resulting parameter estimates. In contrast to most existing results, our consistency analysis does not rely on restrictive assumptions such as persistent excitation conditions of the trajectory data, making it applicable to stochastic adaptive control settings.
Building on the proposed estimator, we further develop an adaptive control algorithm with an attenuating excitation signal that can effectively combine
adaptive estimation and feedback control. Finally,
we are able to show that the resulting closed-loop system is globally stable and that the system trajectory can track, in a long-run average sense, the reference trajectory generated with the true system parameters. The proposed methods and theoretical results are finally validated through simulations in two nonlinear interaction network applications.      
		\end{abstract}
	\end{frontmatter}
	
	\section{INTRODUCTION}
This paper considers the adaptive control problem for a class of nonlinear stochastic systems in the form
\begin{equation}\label{system}
\begin{aligned}
\begin{array}{ll}
			x_{t+1}&=f(A^{*}x_{t}+B^{*}u_{t})+w_{t+1}
            ,\; t\geq 0,
		\end{array}
\end{aligned}
 \end{equation}
where $x_t\in \mathbb{R}^n$
is the system state with initial condition $x_0 \in \mathbb{R}^n$, $u_t\in \mathbb{R}^m$ and $w_{t+1}\in \mathbb{R}^n$  represent the input and random noise process, respectively,  $f:\mathbb{R}^{n}\to \mathbb{R}^{n}$ is a known nonlinear map, and matrices $A^{*}\in \mathbb{R}^{n\times n}$ and $B^{*}\in \mathbb{R}^{n\times m}$ are unknown parameters. This class of nonlinear dynamics occurs in many emerging applications, including, e.g., 
Elman recurrent neural networks (RNNs) in machine learning. Therein, the  model \eqref{system} is also known as a {\it nonlinear recurrent dynamics} with the activation map $f(\cdot)$,
and common choices for the nonlinear activation maps include the Heaviside, ReLU, and logistic functions (see, e.g., \citet{Mienye:2024,rnn, elman:1990}). Another example of the nonlinear model \eqref{system} arises in the dynamics of social networks, where the nonlinear map $f$ reflects the saturation and binarization of opinions as individuals interact with their social peers, while the input $u_t$ models external interventions from influential actors or opinion leaders in the network
(see, e.g., \cite{ dur2010, su, ancona})
Additional examples of the nonlinear model \eqref{system} include neuronal dynamics in biomedical engineering (see \cite{jd1989} and Sect. 2 in \cite{wu2001introduction}) and overflow oscillations of digital filters in signal processing (see, e.g., \cite{Ebert1969, Ooba2003}). 

Adaptive control is a framework for designing feedback controllers that continuously learn from and adjust to unknown or time-varying dynamics online and update their control actions in real time. The primary objective of adaptive control is to ensure closed-loop stability and optimal performance in the presence of uncertainties, variations, and disturbances.
A key idea to adaptive control is to use the
{\it certainty-equivalence principle}, which estimates the unknown parameters online using observed data and designs the controller by treating the latest estimates as the true parameters. Although
 substantial theoretical advances and practical implementations have been achieved for linear stochastic adaptive control \citep{goodwin, kumar, g1995, g1999, liunian}, these results do not extend easily to nonlinear stochastic systems. This is due to several fundamental challenges: first, estimating unknown parameters in a nonlinear model generally leads to nonconvex optimization problems; second, unlike linear systems, the stabilization of nonlinear systems has no ``one-size-fits-all" method, and their behavior can vary drastically depending on the specific nonlinearities involved, such as multiple equilibria, finite escape time, limit cycles, and even bifurcations; finally, nonlinear operations distort the distribution of the process noise and make the joint distribution of the state trajectory nontrivial, which significantly complicates system identification and adaptive control for the nonlinear dynamics \citep{Lennart Ljung}. Motivated by these challenges, this paper investigates the adaptive control for an important class of nonlinearly parameterized stochastic systems given in \eqref{system} when the parameter matrices $A^{*}$ and $B^{*}$ are unknown. 
\subsection{Related literature}
Over the past decades, adaptive control of linear stochastic systems, i.e., the case $f(x)=x$ in \eqref{system}, has been extensively studied. 
Following the landmark work in \cite{str}, rigorous convergence theories for self-tuning regulators have been established by combining least-squares (LS) estimation with minimum-variance control 
(see, e.g., \cite{kumar, g1995}).
In addition, based on the self-convergence property of the weighted least-squares (WLS) algorithm, the convergence theories of adaptive linear–quadratic–Gaussian (LQG) control have been developed in \cite{lg1996, g1999}. More recent advances can be found in~\cite{jiang, Dorfler, liunian, Chiuso} and the references therein. 
Moreover, substantial progress has also been made on adaptive control for linear systems with nonlinear observations $y_t=h(x_t)$, commonly referred to as Wiener systems. Under the assumption that the system matrix $A^{*}$ is stable, several adaptive regulation and tracking methods have been developed (see, e.g., \cite{yanlong2013, lan2023, ZZ2022, bo, wenxiao2023} and the references therein). 

Although extensive theories have been developed for linear stochastic adaptive control, extending these results to nonlinear stochastic systems faces intrinsic challenges. Existing studies largely focus on nonlinear stochastic systems with linearly parameterized uncertainties of the form $x_{t+1}=\theta^{*\top}f(x_{t}, u_{t})+w_{t+1}$, where parameter estimation can still be performed using linear techniques such as LS. However, even when estimation is tractable, ensuring closed-loop stability is considerably more challenging due to nonlinear dynamics, especially when $f$ exhibits superlinear growth.
For example, when $f(x)=O(\|x\|^{b})$ and the unknown parameter $\theta^{*}$ is scalar, \cite{g1997} have shown that adaptive control systems based on LS are globally stabilizable if $b<4$ but unstabilizable when $b\ge 4$. The adaptive stabilization problem based on LS with multiple unknown parameters has later been investigated in~\cite{li2013}. More recent related works can be found, for example, in~\cite{wuquan, liu2025}.

The aforementioned works are primarily concerned with linear-in-parameter systems. In contrast, for nonlinearly parameterized systems, the estimation loss is generally highly nonconvex. As a result, even constructing a consistent parameter estimator is challenging, let alone establishing developing nonlinear adaptive controllers.
Recently, motivated by neural-network–type parameterizations, a growing body of work has investigated parameter estimation for nonlinearly parameterized systems \eqref{system} under a class of nonconvex objectives (see, e.g., \citep{foster, oymak, Kowshik, sattar, lan2025}). These studies typically analyze estimation error using data collected from trajectories of an already stabilized dynamical system. Such trajectory data can often be shown to satisfy persistence of excitation(PE) conditions, which are necessary to guarantee consistent parameter estimation. However, adaptive control is inherently more challenging than purely parameter estimation  for the nonlinearly parameterized system \eqref{system}, which remains unresolved, especially when the closed-loop trajectory fails to satisfy the PE conditions.
Due to these analytical challenges, only a few works have been devoted to adaptive control for nonlinearly parameterized stochastic systems. \cite{hongbin} has studied the adaptive bounded-input–bounded-output (BIBO) stabilization problem using a weighted least-squares (WLS)-type algorithm, where the unknown component belongs to a given set which has a finite number of members.
Furthermore, \cite{li2015} has studied the adaptive stabilization problem for a class of nonlinearly parameterized uncertain systems with a scalar-valued unknown parameter. In the scalar case, they showed that the required PE conditions for parameter estimation actually hold. However, as pointed out by the authors in \cite{li2015}, ``for the multiple-parameter case, the analysis will in general be extremely difficult".
\subsection{Main contributions}
This paper studies online identification and adaptive control for the nonlinearly parameterized stochastic system \eqref{system} with an unknown multidimensional parameter $[A,B]\in\mathbb{R}^{n\times(n+m)}$.  Inspired by \cite{lg1996, g1999}, we propose a novel nonlinear WLS-type parameter estimation algorithm and develop an adaptive control method that incorporates an attenuating exploration signal into the certainty-equivalence controller. The main contributions are threefold:
\begin{itemize}
    \item We propose a novel online parameter estimation algorithm for recovering the parameter $[A^{*}, B^{*}]$ from the trajectory data $\{x_t, u_t\}_{t\geq 0}$.
Unlike the parameter estimation results reported in \citep{foster, Kowshik, oymak, sattar}  
for system \eqref{system}, our estimation algorithm achieves strong consistency without requiring the data to be  persistently exciting (See
Remark \ref{re1}).
    \item We design an adaptive controller and show that, under a stabilizability condition on the system, the closed-loop system induced by the proposed adaptive controller is globally stable.  In particular, global stability is also guaranteed when the parameter estimates do not converge to the true parameter. Compared with existing results,
    our results are applicable to a class of nonlinear stochastic systems with multiple unknown parameters, which are not linearly parameterized.
    \item We further prove that, if the exploration signal introduced in the adaptive controller provides sufficient excitation, then the parameter estimates are strongly consistent along the adaptive closed-loop trajectory. At the same time, the closed-loop trajectory generated by the proposed adaptive controller tracks the reference trajectory induced by the corresponding controller designed with full knowledge of the true parameter.
    This adaptive tracking problem has received relatively limited attention in the existing literature on adaptive
control of nonlinear stochastic systems.
\end{itemize}

The remainder of the paper is organized as follows. Section \ref{sec2} provides the preliminaries and formulates the problem. Section \ref{sec3} presents the results on parameter estimation, including the proposed online parameter estimation algorithm and its convergence properties. Section \ref{sec4} addresses the adaptive control problem, focusing on closed-loop stability under adaptive nonlinear control and the resulting tracking performance. Section \ref{sec5} illustrates two application scenarios with numerical simulations. Finally, Section \ref{sec6} concludes the paper.

\noindent {\bf Notations.} We use $\|\cdot\|$ to denote the Euclidean norm for vectors and the induced operator norm for matrices. For matrices $M, N\in\mathbb{R}^{p\times q}$, the Euclidean/Frobenius inner product is defined as
$\langle M, N\rangle = \operatorname{tr}(M^{\top}N)$.
For a square symmetric matrix $M$, we denote its largest and smallest eigenvalues by $\lambda_{\max}(M)$ and $\lambda_{\min}(M)$, respectively, and use $\operatorname{tr}[M]$ to denote the trace of $M$.
We use the shorthand notation $[n]=\{1,2,\dots,n\}$ for any positive integer $n$. Finally, $I_n$ denotes the $n\times n$ identity matrix.
\\ \noindent {\bf Regularity of Functions.} 
A mapping $f: \mathbb{R}^{n}\to \mathbb{R}^{n}$ is said to be 
 {\it $\alpha-$strongly monotone} on 
$K$ if there exists a constant $\alpha>0$ such that $(x-y)^{\top}(f(x)-f(y))\geq \alpha\|x-y\|^{2}, \; \forall x, y\in K.$
The mapping $f$ is said to be \emph{Lipschitz continuous} on a set $K$ if there exists a constant $\nu>0$ such that $\|f(x)-f(y)\|\leq \nu \|x-y\|,\; \forall x, y\in K.$
In this case, we also say that $f$ is \emph{$\nu$-Lipschitz continuous}.

\section{Preliminaries and Problem Formulation}\label{sec2}
In this paper, we consider the stochastic nonlinear system \eqref{system}. The matrices $A^{*}$ and $B^{*}$ are unknown, while the nonlinear function $f(\cdot)$ is known and the trajectory data $\{x_\tau, u_\tau\}_{0 \le \tau \le t}$ are observable at each time $t\geq 0$. 
Let $\{\mathcal{F}_{t}\}$ be a filtration, i.e., a nondecreasing family of $\sigma$-algebras, defined by $\mathcal{F}_{t}=\sigma(x_0, u_i, w_{i}, 0\leq i \leq t).$
We assume that the nonlinear function $f(\cdot)$ and noise process $\{w_t\}_{t\geq 0}$ satisfy the following two assumptions.
\begin{assumption}\label{assum_nonlinear}
There exist a nonincreasing function $\alpha: \mathbb{R}_{\ge 0} \to \mathbb{R}_{>0}$ and a nondecreasing function $\beta: \mathbb{R}_{\ge 0} \to \mathbb{R}_{>0}$ with $\inf_{z \in \mathbb{R}^n} \alpha(\|f(z)\|) > 0$ and  $\sup_{z \in \mathbb{R}^n} \beta(\|f(z)\|) < \infty$, such that, for any $c\geq 0$ and all 
$\|x\|, \|y\|\leq c,$
\begin{equation}\label{e3}
\begin{aligned}
&\|f(x)-f(y)\|\leq \beta(c)\|x-y\|,\; \text{and} \\
(x&-y)^{\top}(f(x)-f(y))\geq \alpha(c)\|x-y\|^{2}.
\end{aligned}
\end{equation}
\end{assumption}
\begin{assumption}\label{assum_noise}
 The sequence $\{w_{t}, \mathcal{F}_{t}\}_{t\geq 0}$ is a martingale difference sequence, and there exist constants $\gamma>2$ and $\sigma>0$ such that, for all $t\ge 0$,
\begin{equation}\nonumber
    \sup_{t\geq 0}
    \mathbb E \left[\left\|w_{t+1}\right\|^{\gamma}\mid \mathcal{F}_t\right]<\infty,\; 
    \mathbb E \left[w_{t+1}w_{t+1}^{\top}\mid \mathcal{F}_t\right]\geq \sigma I_{n}.
\end{equation}
Moreover, if $\sup_{x\in\mathbb{R}^n}\|f(x)\|<\infty$, we additionally assume that $\{w_t\}_{t\geq 0}$ is almost surely bounded.
\end{assumption}
Next, two
sufficient conditions for Assumption \ref{assum_nonlinear} are given in Lemma \ref{lemsec2} and
further illustrated by the examples given in Subsection
2.1, which shows that this assumption holds in a wide range of stochastic nonlinear systems. Note that nonlinear systems in which the noise term $w_{t+1}$ enters inside the nonlinearity $f(\cdot)$ can also be transformed into the form \eqref{system} when $f(\cdot)$ is monotone and bounded and $\{w_t\}$ is i.i.d.\ Gaussian; see Example~\ref{ex1}.
\begin{lemma}\label{lemsec2}
Suppose $f$ is a component-wise mapping of the form $f(z)=(f_{1}(z_1),\cdots,f_{n}(z_n))^{\top}$, where each $f_{i}(\cdot):\mathbb{R}\rightarrow \mathbb{R}$ acts on the $i$-th component of $z$. Then Assumption~\ref{assum_nonlinear} holds in either of the following two cases:
\begin{itemize}
\item Case (i): For each $i\in [n]$, the function $f_i$ is $\alpha_i$-strongly monotone and $\beta_i$-Lipschitz continuous on $\mathbb{R}$, for some constants $\alpha_i>0$ and $\beta_i>0$. Then $f$ satisfies Assumption \ref{assum_nonlinear} with $\alpha(\cdot)\equiv \min_{i\in [n]}\alpha_i$ and $\beta(\cdot)\equiv \max_{i\in [n]}\beta_i$;
\item Case (ii): For each $i\in [n]$, the function $f_i$ is bounded, differentiable, and for any $c\geq 0$, 
\begin{equation}\label{e4}
\inf_{i\in [n], |x|\leq c} \frac{d}{dx}f_{i}(x)>0,\; \sup_{i\in [n], |x|\leq c} \frac{d}{dx}f_{i}(x)<\infty.
\end{equation}
In this case, $f$ satisfies Assumption~\ref{assum_nonlinear} with
\begin{equation}\label{e5}
\begin{aligned}
\alpha(c)&=\inf\limits_{i\in [n], |x|\leq c} \frac{d}{dx}f_{i}(x), c\geq 0, \;\text{and}\\
\beta(c)&=\sup\limits_{i\in [n], |x|\leq c} \frac{d}{dx}f_{i}(x), c\geq 0.
\end{aligned}
\end{equation}
\end{itemize}
\end{lemma}
The proof of Lemma~\ref{lemsec2} is given in Appendix~\ref{ap6}.
\subsection{Examples}\label{subsec:examples}
In the following, we present three representative examples of the nonlinear system~\eqref{system} in which $f$ satisfies Assumption~\ref{assum_nonlinear}.

\begin{example}[Nonlinear Opinion Dynamics]\label{ex3}
Consider a nonlinear opinion dynamical system with external noise
 (see, e.g., \citet{dur2010, ancona,Bizyaeva} for background on nonlinear opinion formation)
\begin{equation}\label{example1}
x_{t+1} = \,a\tanh(A^{*}x_{t}+B^{*}u_{t}
)+w_{t+1},
\end{equation}
where $x_t=[x_{1,t},\cdots, x_{n,t}]\in \mathbb{R}^{n}$ represents the opinion states of the $n$ agents at time $t$; $A^{*}\in\mathbb{R}^{n\times n}$ and $B^{*}\in\mathbb{R}^{n}$ denote the unknown adjacency matrices and the connections of influencers of the underlying social network;
$a>0$ is the maximum sensitivity of agents to social influence; $u_{t}\in \mathbb{R}^m$ is a pinning (or influence) input, potentially generated by a leader, that drives the opinion dynamics of the entire network. 
The hyperbolic tangent function $tanh(\cdot)$ is defined as $\tanh(z)=\frac{e^{z}-e^{-z}}{e^{z}+e^{-z}}, \; z\in \mathbb{R},$
and is applied elementwise to vectors in \eqref{example1}. 
This nonlinear function captures the saturation effects in opinion updates arising from interactions among social peers.
  By Case (2) of Lemma \ref{lemsec2}, the function $a\tanh(\cdot)$ satisfies Assumption~\ref{assum_nonlinear} with the following functions $\alpha(\cdot)$ and $\beta(\cdot)$:
 \begin{align*}\nonumber
 \alpha(x)&=a\inf_{|z|\leq x}\frac{d}{dz}tanh(z)=\frac{4a}{(e^{x}+e^{-x})^{2}},\\
 \beta(x)&=a\sup_{|z|\leq x}\frac{d}{dz}tanh(z)=a,\;\forall x\geq 0.
 \end{align*}
\end{example}
\begin{example}[RNN Dynamics]\label{ex2}
Consider the classical Elman recurrent neural networks (RNN) state equation with external
noise (See., e.g., Equation $(14)$ of \cite{Mienye:2024}, p.~634 of \cite{rnn}, \cite{elman:1990}),
\begin{equation}\label{e14}
x_{t+1}=\sigma(A^{*}x_{t}+B^{*}u_{t})+w_{t+1},
\end{equation}
where $x_t$ and $u_t$ are
the state vector and 
the input data at timestamp $t$. $A^{*}\in \mathbb{R}^{n\times n}$ and $B^{*}\in \mathbb{R}^{n\times m}$ are unknown state weight matrix and input weight matrix to be trained. The nonlinear function $\sigma:\mathbb{R}^{n}\to\mathbb{R}^{n}$ is an activation function of the form $\sigma(\tau_{1},\dots,\tau_{n})=[\sigma_{1}(\tau_{1}),\dots,\sigma_{n}(\tau_{n})]^{\top}$.
Classical examples of activation functions include: 

1) The sigmoid function, i.e., $\sigma_{i}(z)=\frac{1}{1+e^{-z}}, \forall z\in \mathbb{R}, i\in [n].$ In this case, by Case (2) of Lemma \ref{lemsec2}, Assumption~\ref{assum_nonlinear} is satisfied with $\alpha(x)=\frac{e^{-x}}{(1+e^{-x})^{2}}$ and $\beta(x)=\frac{1}{4}, \forall x\in \mathbb{R}$. 

2) The leaky ReLU (LReLu) function, i.e., $\sigma_{i}(z)=\mathrm{LReLU}(z)=\max(b z, z),\; b \in (0,1), \forall z\in \mathbb{R}, i\in [n]$. In this case, by Case (1) of Lemma \ref{lemsec2}, Assumption~\ref{assum_nonlinear} holds with  constant functions $\alpha(x)\equiv b$ and $\beta(x)\equiv 1, \forall x\in \mathbb{R}$.
\end{example}
\begin{example}[Neuronal Dynamics]\label{ex1}
To study neuronal behaviors, the following neuronal dynamics model has been proposed and investigated (see, e.g., Equation $(1)$ in \cite{yx2022}, \cite{jd1989}, Sect. 2 in \cite{wu2001introduction}):
\begin{equation}\label{e17}
x_{i,t+1}= I\!\left(\sum_{j=1}^{n}a_{i,j}^{*}x_{j,t}+c_{i,t}+\eta_{i,t}>0\right), i\in [n], t\geq 0,
\end{equation}
where $x_{i,t}\in \{0,1\}$ denotes the state of neuron $i\in [n]$ at time $t$, indicating whether neuron $i$ is activated, and can be detected through experimental methods; the function $I(\cdot>0)$ represents the McCulloch–Pitts neuron model, defined as $I(x>0)=1$ if $x> 0$ and $I(x>0)=0$ otherwise; the parameter $a_{ij}^{*}= \gamma_{i}t_{ij}/n$ represents the total interaction strength between agents $i$ and $j$, determined by the synaptic strength $\gamma_{i}$ and the link strength $t_{ij}$ between them. The term $c_{i,t}$ represents a bias or an external intervention applied to neuron $i$ at time $t$.
 The noise term $\eta_{i,t}$ denotes the disturbance affecting neuron $i$ at time $t$,
and is assumed to be independent and identically distributed (i.i.d.) standard Gaussian. 

For model \eqref{e17}, note that $\mathbb{E}[x_{i,t+1}\mid \mathcal{F}_{t}]=1-F(-\sum_{j=1}^{n}a_{i,j}^{*}x_{j,t}+c_{i,t}),$ where $F(\cdot)$ denotes the cumulative distribution function (cdf) of the standard Gaussian distribution. Let $w_{i,t+1}=x_{i,t+1}-\mathbb{E}[x_{i,t+1}\mid \mathcal{F}_{t}]$, then $\{w_{i,t}, \mathcal{F}_t\}$ forms a martingale difference sequence. Consequently, system~\eqref{e17} can be rewritten as
\[
x_{i,t+1} = 1-F\!\left(-\sum_{j=1}^{n}a_{i,j}^{*}x_{j,t}+c_{i,t}\right)+w_{i,t+1}.
\]
By Case (2) of Lemma \ref{lemsec2}, one can verify that the function $1-F(\cdot)$ satisfies Assumption \ref{assum_nonlinear} with functions $\alpha(x)=\frac{1}{\sqrt{2\pi}}exp(-\frac{1}{2}x^{2})$ and $\beta(x)=\frac{1}{\sqrt{2\pi}}, \forall x\in \mathbb{R}.$
\end{example}


\subsection{Problem formulation}
Denote the system parameter by $\theta^{*}=\left[A^{*}, B^{*}\right]^{\top}\in \mathbb{R}^{(n+m)\times n}$, and assume that the true parameter $\theta^{*}\in \Theta$, where $\Theta$ is a compact convex set in $\mathbb{R}^{(n+m)\times n}$. We first consider the 
parameter estimation problem for the unknown 
parameter $\theta^{*}$.
\begin{problem}
At each time $t\geq 0$, given the trajectory observations $\{x_\tau, u_{\tau-1}\}_{1\leq \tau\leq t}$, construct an online estimate $\hat \theta_t$ of the unknown parameter $\theta^*$, and determine conditions on the trajectory data $\{x_\tau, u_{\tau-1}\}
$ under which
\begin{equation*}
\lim_{t\to \infty}\left\|\hat{\theta}_{t}- \theta^* \right\|^{2}=0,\, a.s.
\end{equation*}
\end{problem}
With the estimation algorithm provided in Problem 1, we further solve the adaptive control problem. 

Let the space of admissible feedback control laws be $$\mathcal U =\{\mathbf u: \mathbb R^n \to \mathbb R^m \mid \mathbf u\; \text{ is Lipschitz continuous}\}.$$   
Let $\pi_{\cdot}$ be a given control design mechanism, which defines a mapping from $\Theta$ to $\mathcal U$, and let $\{x_{t}^{*}\}$ denote the corresponding reference state trajectory generated by system \eqref{system} under the  feedback controller $u_t^{*}=\pi_{\theta^{*}}(x_t)$, i.e.,
\begin{equation}\label{reference_trajectory}
x_{t+1}^{*}=f(A^{*}x_{t}^{*}+B^{*}\pi_{\theta^{*}}(x_{t}^{*}))+w_{t+1},\;x_{0}^{*}=x_0,\;\text{a.s.}
\end{equation}
We will study the following adaptive control problem.
\begin{problem}
Given a control design mechanism $\pi_{\cdot}$, without knowing the true parameter $\theta^{*}$, design an adaptive feedback control algorithm together with the online parameter estimation algorithm 
   \begin{align*}
    u_t &= \mathcal G (x_t, \hat \theta_t),\\
    \hat \theta_t &=\mathcal H (\{u_\tau\}_{\tau\leq t-1}, \{x_\tau\}_{\tau\leq t}),\quad\forall t\ge 0,
   \end{align*}
where $\{x_t\}$ is the closed-loop trajectory of system \eqref{system} with adaptive feedback control $\{u_t\}$, such that the resulting trajectory $\{x_t, u_t\}$ starting from any initial
state $x_0$ has the following two properties:
\begin{itemize}
\item Adaptive stabilization: the closed-loop system is globally stable, i.e.,
$$\limsup_{t\to \infty} \frac{1}{t}\sum_{\tau=0}^{t}\|x_\tau\|^{2} <\infty,\;\text{a.s.}$$
\item Asymptotic tracking: the averaged
tracking error converges to zero almost surely, i.e.,
\begin{equation*}
     \lim_{t\to \infty}\frac{1}{t}\sum_{\tau=0}^{t}\left(\left\|x_{\tau}-x_{\tau}^{*}\|^{2}+ \|u_{\tau}- u_{\tau}^{*}\right\|^{2}\right)=0,\;\text{a.s}.,
    \end{equation*}
    where $\{x_{t}^{*}\}_{t\geq 0}$ and $\{u_{t}^{*}\}_{t\geq 0}$ are given in \eqref{reference_trajectory}.
\end{itemize}
\end{problem}
These two problems outlined above will be addressed in Section \ref{sec3} and Section \ref{sec4}, respectively.
\section{Parameter estimation}\label{sec3}
In this section, we  study the estimation problem of the unknown system parameter $\theta^*=(A^{*}, B^*)$. The following assumption on $\theta^{*}$ is needed.
\begin{assumption}\label{assum_parameter}
There exists a constant $\epsilon>0$ such that 
$\{\theta\in\mathbb{R}^{(n+m)\times n}:\|\theta-\theta^{*}\|\le \epsilon\}$ is contained in a known compact convex set $\Theta$.
\end{assumption}
In the following, we construct the parameter estimation algorithm and present the corresponding estimation error analysis in Sections \ref{sec3.1} and \ref{sec3.2}, respectively.
\subsection{Parameter estimation algorithm}\label{sec3.1}
Denote by $\phi_{t}=\left[x_{t}^{\top}, u_{t}^{\top}\right]^{\top} \in \mathbb{R}^{n+m}$ the regression vector. To construct an online estimation algorithm, for each $t \ge 0$, we consider the weighted least-squares loss function given by
\begin{equation*}
S_{t}(\theta)=\frac{1}{\mu_{t}}\mathbb{E}\left[\left\|x_{t+1}-f(\theta^{\top}\phi_{t})\right\|^{2}\mid \mathcal{F}_t\right],\;\theta\in \Theta,
\end{equation*}
where $\{\mu_{t}\}_{t\geq 0}$ is a weight sequence to be designed and is assumed to be adapted to the filtration $\{\mathcal{F}_t\}_{t\geq 0}$. The loss function is minimized at $\theta^*$, and $\min_{\theta\in\Theta} S_t(\theta)= S_t(\theta^*)=\mu_{t}^{-1}\mathbb{E}[\left\|w_{t+1}\right\|^{2}\mid \mathcal{F}_t], \forall t\geq 0$. 

To estimate $\theta^{*}$ by minimizing $S_t(\theta)$, we develop an online WLS-type algorithm, summarized in  Algorithm~1.
\begin{algorithm}[h]\label{alg11}
 \caption{Online WLS-type estimator for $\theta^{*}$}
{\bf Initialization}:  
Set the initial estimate as $\hat{\theta}_{0}=[\hat{A}_{0}, \hat{B}_{0}]^{\top}\in \Theta$,  
initialize the matrix $P_{0}=I_{n+m}$, and choose a positive real $\delta>0$.  
\\{\bf Update flows}: Given the observation $(u_t, x_{t+1})$ for each time $t>0$, the estimate  $\hat{\theta}_{t+1}$ is recursively updated as follows: 
 \begin{align}
&\hat{\theta}_{t+1}=\Pi_{\Theta}^{t+1}\{\hat{\theta}_{t}+\frac{1}{\mu_{t}}d_{t}P_{t+1}\phi_{t}[x_{t+1}-f(\hat{\theta}_{t}^{\top}\phi_{t})]^{\top}\},\label{sec3_eq9}\\
 &P_{t+1}=P_{t}-a_{t}d_{t}^{2}P_{t}\phi_{t}\phi_{t}^{\top}P_{t},\label{eq6}\\
&a_{t}=\left(\mu_{t}+d_{t}^{2}\phi_{t}^{\top}P_{t}\phi_{t}\right)^{-1},\label{sec3_alg11}
 \end{align}
 where $ \phi_{t}=\left[x_{t}^{\top}, u_{t}^{\top}\right]^{\top}, \mu_{t}=(1+\log r_{t})^{1+\delta}+d_{t}\overline{g}_{t}^{2}\phi_{t}^{\top}P_{t}\phi_{t}, r_{t}=1+\sum_{\tau=0}^{t}\|\phi_{\tau}\|^{2},$ 
 $d_{t}=\frac{1}{2}\alpha(\|\hat{\theta}_{t}^{\top}\phi_{t}\|+\max\limits_{\theta\in \Theta}\|\theta^{\top}\phi_{t}\|),$ $\overline{g}_{t}=\beta(\|\hat{\theta}_{t}^{\top}\phi_{t}\|+\max\limits_{\theta\in \Theta}\|\theta^{\top}\phi_{t}\|),$ and functions $\alpha$ and $\beta$ are defined in Assumption \ref{assum_nonlinear}. The projection operator $\Pi_{\Theta}^{t+1}$ is defined as $$\Pi_{\Theta}^{t+1}(x)=
\begin{cases}
x, & \text{if } x \in \Theta, \\[4pt]
\argmin\limits_{y\in \Theta_{\epsilon}}\operatorname{tr}[(x-y)^{\top}P_{t+1}^{-1}(x-y)], & \text{if } x \notin \Theta,
\end{cases}
$$
where $\Theta_{\epsilon}=\{\theta\in \Theta: \|\theta-\theta^{*}\|\leq \epsilon\}$, and $\epsilon$ is given in Assumption \ref{assum_parameter}.
 \end{algorithm}
 
We now give an intuitive and constructive interpretation of Algorithm $1$. Since $\phi_{t}$ and $\mu_{t}$ are $\mathcal{F}_t-$measurable and $\mathbb{E}[w_{t+1}\mid \mathcal{F}_t]=0$, it follows that
\begin{equation}\label{sec3_eq9}
\begin{aligned}
S_{t}(\theta)-S_{t}(\theta^*)
=\mu_{t}^{-1}\|f(\theta^{*\top}\phi_{t})-f(\theta^{\top}\phi_{t})\|^{2}.
\end{aligned}
\end{equation}
By Assumption 1, the nonlinear map $f(\cdot)$ satisfies
\begin{equation}\label{sec3_eq10}
\begin{aligned}
&d_t\|\theta^{*\top}\phi_{t}-\theta^{\top}\phi_{t}\|^{2}\\
\leq &(\theta^{*\top}\phi_{t}-\theta^{\top}\phi_{t})^{\top}[f(\theta^{*\top}\phi_{t})-f(\theta^{\top}\phi_{t})],
\end{aligned}
\end{equation}
where $d_t$ depends on the monotonicity function $\alpha(\cdot)$ and is specified in Algorithm 1. 
On the other hand, applying the Cauchy–Schwarz inequality yields
\begin{equation}\label{sec3_eq11}
\begin{aligned}
&(\theta^{*\top}\phi_{t}-\theta^{\top}\phi_{t})^{\top}[f(\theta^{*\top}\phi_{t})-f(\theta^{\top}\phi_{t})]\\
\leq &\|\theta^{*\top}\phi_{t}-\theta^{\top}\phi_{t}\|\|f(\theta^{*\top}\phi_{t})-f(\theta^{\top}\phi_{t})\|.
\end{aligned}
\end{equation}
Squaring both sides of \eqref{sec3_eq11} and using \eqref{sec3_eq10}, we obtain
\begin{equation}\label{sec3-eq12}
\begin{aligned}
&\|f(\theta^{*\top}\phi_{t})-f(\theta^{\top}\phi_{t})\|^{2}\\
\geq &d_t(\theta^{*\top}\phi_{t}-\theta^{\top}\phi_{t})^{\top}[f(\theta^{*\top}\phi_{t})-f(\theta^{\top}\phi_{t})]
\\= &\langle \theta^{*}-\theta,  
d_{t}\phi_{t}[f(\theta^{*\top}\phi_{t})-f(\hat{\theta}_{t}^{\top}\phi_{t})]^{\top}\rangle.
\end{aligned}
\end{equation}
Then, combining \eqref{sec3_eq9} and \eqref{sec3-eq12}, we have
\begin{equation}\label{sec3_eq14}
\begin{aligned}
&S_{t}(\theta)-S_{t}(\theta^*)\\
\geq &\mathbb{E}\left[\langle \theta-\theta^{*}, -\mu_{t}^{-1}d_{t}\phi_{t}[x_{t+1}-f(\hat{\theta}_{t}^{\top}\phi_{t})]^{\top}\rangle\mid \mathcal{F}_t\right].
\end{aligned}
\end{equation}
Since $f(\cdot)$ may be non-differentiable, \eqref{sec3_eq14} motivates using $\mu_{t}^{-1}d_{t}\phi_{t}[x_{t+1}-f(\hat{\theta}_{t}^{\top}\phi_{t})]^{\top}$ as a gradient-like surrogate in the parameter update. Instead of directly forming a first-order gradient algorithm, we adopt a Newton-type second-order update, which relies on multiplying an inverse (or an approximation) of the Hessian.
Using the well-known  Sherman–Morrison formula \citep{hager}, 
\[
P_{t+1}^{-1}=P_{t}^{-1}+\frac{d_{t}^{2}}{\mu_t}\phi_{t}\phi_{t}^{\top}=P_{0}^{-1}+\sum_{\tau=0}^{t}\frac{d_{\tau}^{2}}{\mu_\tau}\phi_{\tau}\phi_{\tau}^{\top},
\]
the matrix $P_{t+1}$ provides a recursive Gauss–Newton type approximation of the inverse Hessian of the WLS loss
 $S_t(\theta)$. Accordingly, Algorithm 1 uses $\mu_{t}^{-1}d_{t}P_{t+1}\phi_{t}[x_{t+1}-f(\hat{\theta}_{t}^{\top}\phi_{t})]^{\top}$ as the update term for the online parameter estimate. The weight sequence $\{\mu_t\}$ is designed, following the idea in \cite{lg1996}, to ensure the self-convergence property of the algorithm (see Conclusion (i) of Theorem 1), which plays a key role in adaptive control design.

Since a Newton-type update step alone may produce an estimate outside the compact set $\Theta$, we introduce a projection operator in the algorithm. This time-varying projection is performed with respect to the matrix  $P_{t+1}^{-1}$ and the set $\Theta_{\epsilon}$. The reason for using this projection is to guarantee that it is invoked only finitely many times, that is, there exists a constant $T>0$ such that $\hat{\theta}_{t+1}=\hat{\theta}_{t}+\mu_{t}^{-1}d_{t}P_{t+1}\phi_{t}[x_{t+1}-f(\hat{\theta}_{t}^{\top}\phi_{t})]^{\top}, \forall t>T$. This property also follows from the specific choice of the weight sequence $\{\mu_t\}$. Details of this argument can be found in the proof of Theorem~\ref{thm1}. 
\subsection{Convergence analysis}\label{sec3.2}
In this subsection, we establish the strong consistency of the parameter estimates given by Algorithm $1$.
\begin{theorem}\label{thm1}
Suppose Assumptions $\ref{assum_nonlinear}$-$\ref{assum_parameter}$ hold. If the control sequence $\{u_t\}_{t\geq 0}$ satisfies 
that there exist constants $c_{1}>0$ and $c_{2}>0$ such that $\|u_{t}\|\leq c_{1}\|x_{t}\|+c_{2}$ for all $t\geq 0$, then the parameter estimates produced by Algorithm $1$ satisfy the following three properties:
\begin{itemize}
\item[(i)] There exists a real-valued random matrix $\bar{\theta}$, such that
 \begin{equation}\label{eqq9}
 \lim_{t\to \infty}\hat{\theta}_{t}=\bar{\theta},\;\;a.s.
 \end{equation}
\item[(ii)] 
The adaptive prediction error has the following asymptotic upper bound as $t\to \infty$:
\begin{equation}
\sum_{\tau=1}^{t}\mu_{\tau}^{-1}\left\|f(\theta^{*\top}\phi_{\tau})-f(\hat{\theta}_{\tau}^{\top}\phi_{\tau})\right\|^{2}=O(1),\;a.s.
\end{equation}
  \item[(iii)] The parameter estimation error has the following asymptotic upper bound as $t\to \infty$:
\begin{equation}\label{6}
\left\|\theta^{*}-\hat{\theta}_{t}\right\|^{2}=O\left(\frac{(\log r_t)^{1+\delta}}{\lambda_{t}}\right), a.s.,
\end{equation}
where $r_{t}$ and $\mu_t$ are defined in Algorithm 1, and $\lambda_t$ is given by
\begin{equation}\label{lam}
\lambda_{t}=\lambda_{\min}\left\{\sum\limits_{\tau=0}^{t}\frac{\phi_{\tau}\phi_{\tau}^{\top}}{1+\|\phi_{\tau}\|^{2}}\right\}.
\end{equation}
\end{itemize}
\end{theorem}
The proof of Theorem \ref{thm1} is provided in Appendix \ref{ap2}.

Notice that the condition imposed on the control sequence $\{u_t\}$ in Theorem~\ref {thm1} is not restrictive, which holds for any state feedback controller that is Lipschitz in the state. 
Specifically, Conclusion (i) of Theorem~\ref {thm1} ensures that the algorithm converges to a fixed point $\bar \theta$, though probably not $\theta^*$, without any excitation condition requirement.
This non-divergence property prevents the parameter estimates from varying excessively over time and the divergent propagation of the historical estimate error.
This enables us to use the estimate $\hat\theta_t$ as a surrogate parameter in the adaptive control design, as shown in Section \ref{sec4}. Furthermore, Conclusion~(ii) of Theorem~\ref{thm1} establishes an asymptotic upper bound on the accumulated prediction regret, which is instrumental in establishing the global stability of the adaptive closed-loop system, as shown in Theorem~\ref{thm3}.
\begin{remark}\label{re1}
From Conclusion~(iii), the parameter estimation error converges to zero almost surely provided that, almost surely,
\begin{equation}\label{data_condition}
\left(\log \left(\sum_{\tau=0}^{t}\|\phi_{\tau}\|^{2}\right)\right)^{1+\delta}=o\left(\lambda_{\min}\left\{\sum\limits_{\tau=0}^{t}\frac{\phi_{\tau}\phi_{\tau}^{\top}}{1+\|\phi_{\tau}\|^{2}}\right\}\right).
\end{equation}
In prior work on parameter estimation (e.g., \cite{foster, Kowshik, oymak, sattar}),
the system \eqref{system} is typically assumed to be asymptotically global exponential stable (see Definition $2$ in \cite{foster}), and the inputs and noises $\{u_t\}, \{w_t\}$ are taken to be independent and sub-Gaussian. Under these assumptions, the trajectory $\{\phi_t\}$  satisfies two key properties almost surely: 1) the PE condition, i.e.,  $t=O\left(\lambda_{\min}\left\{\sum\limits_{\tau=0}^{t}\phi_{\tau}\phi_{\tau}^{\top}\right\}\right), a.s.,$ (see, e.g., Theorem $4$ in \cite{foster}), and 2) bounded-moment condition, i.e., $\sum_{\tau=0}^{t}\|\phi_{\tau}\|^{2+\rho}=O(t)$ for some $\rho>0$. Together, these two conditions imply \eqref{data_condition}, and hence Algorithm~1 is strongly consistent under this commonly studied setting.

More importantly, \eqref{data_condition} indicates that strong consistency of Algorithm~1 may still hold even when the classical PE condition fails. For instance, if the trajectory satisfies $\sum_{\tau=0}^{t}\|\phi_{\tau}\|^{2+\rho}=O(t)$and, instead of PE, obeys a weaker non-PE condition $(\log t)^{1+\delta}=o\left(\lambda_{\min}\left\{\sum\limits_{\tau=0}^{t}\phi_{\tau}\phi_{\tau}^{\top}\right\}\right), a.s.$, then \eqref{data_condition} remains valid.
\end{remark}
\section{Nonlinear Adaptive Control}\label{sec4}
In Section~\ref{sec3}, we developed an online parameter estimation algorithm and established the strong consistency of the resulting estimates without requiring the PE conditions for the trajectory data. Building on these results, we investigate the adaptive control problem in this section.

Suppose given a control design mechanism $\pi_{\cdot} : \Theta \to \mathcal{U}$, where $\mathcal{U}$ denotes the space of all admissible feedback control laws defined as
$\mathcal U =\{\mathbf u: \mathbb R^n \to \mathbb R^m : \mathbf u\; \text{is Lipschitz continuous}\}.$ Accordingly, for each $\theta\in \Theta$, the policy $\pi_{\theta}(\cdot)\in \mathcal U$, and we assume that there exists a  constant $L>0$ such that for all $\theta \in \Theta$, $\pi_{\theta}(\cdot)$ is $L-$Lipschitz continuous. 

Let $\{x_{t}^{*}\}_{t\geq 0}$ denote the corresponding reference state trajectory generated by system \eqref{system} under the  feedback controller $u_t=\pi_{\theta^{*}}(x_t)$, i.e.,
\begin{equation}\label{sec4_trajectory}
x_{t+1}^{*}=f(A^{*}x_{t}^{*}+B^{*}\pi_{\theta^{*}}(x_{t}^{*}))+w_{t+1},\;\;\text{a.s.,}
\end{equation}
with initial condition $x_{0}^{*}=x_0$. The problem considered here is the following: without knowing the true parameter $\theta^{*}$, how do we design an adaptive controller such that the resulting closed-loop system is globally stable and the state trajectory tracks the reference trajectory $\{x_{t}^{*}\}$. 

{\bf Adaptive Control Law.} With the parameter estimates produced by Algorithm~1, a intuitive approach to design an adaptive controller is to apply the certainty-equivalence principle (CEP), whereby the unknown true parameter is replaced by its online estimate in the controller. Given a control design mechanism $\pi_{\cdot}$ and the current parameter estimate $\hat \theta_t=[\hat A_t, \hat B_t]$ available up to time t, we construct the adaptive control law as
\begin{equation}\label{control}
u_{t}=\pi_{\hat{\theta}_{t-1}}(x_{t})+v_{t},
\end{equation}
where $v_t=t^{-b}\epsilon_{t}$ for some constant b>0. Here $\{\epsilon_t\}_{t\geq0}$ is an i.i.d. random variable sequence, independent of $\{w_t\}_{t\geq0}$, satisfying
\begin{equation}\label{sec4eq23}
\mathbb{E}\left[\epsilon_{k}\right]=0, \;\;\mathbb{E}\left[\epsilon_{k}\epsilon_{k}^{\top}\right]=I,\;\; \|\epsilon_t\|\leq \bar{\epsilon}
\end{equation}
for some constant $\bar{\epsilon}>0$. The additive perturbation $v_t$ serves as a probing (exploration) signal that injects excitation into the closed-loop, which is essential for consistent parameter learning and, consequently, effective adaptation.

Under system \eqref{system} and the adaptive law \eqref{control}, the resulting closed-loop system is given by
\begin{equation}\label{sec4_eq21}
x_{t+1}=f(A^{*}x_{t}+B^{*}(\pi_{\hat{\theta}_{t-1}}(x_{t})+v_{t}))+w_{t+1}.
\end{equation}
To establish stability of the closed-loop system \eqref{sec4_eq21}, we impose a stabilizability requirement on the control design mechanism $\pi_{\cdot}$. To this end, we first recall the notion of finite-gain $L_q$ stability for nonlinear systems.
\begin{definition}[Finite-gain $L_q$-stability](see, e.g., \cite[Def.~1.2.1]{van})
The stochastic system $x_{t+1}=F(x_{t},v_{t},w_{t})$ with input $v_t\in \mathbb{R}^m$, state $x_t\in \mathbb{R}^n,$ and noise $w_t \in \mathbb{R}^n$ is said to be finite-gain $L_q$-stable (with respect to $\{v_{t}\}$ and $\{w_{t}\}$) for some $q>0$, if there exists a constant $k_{q}>0$ such that for 
every initial condition $x_0\in \mathbb{R}^n$ there exists a constant $b_{q}(x_0)>0$ with the property 
that
\[
\sum_{\tau=0}^{t-1}\|x_{\tau+1}\|^{q} \le k_{q}\sum_{\tau=0}^{t-1}(\|v_\tau\|^{q}+\|w_{\tau+1}\|^{q})+b_{q}(x_0), \text{a.s}.
\]
for any $t\geq 0$ and any random variable sequences $\{v_{t}\}$ and $\{w_{t}\}$.
\end{definition}
Consider now the following stabilizability notion.

{\bf Finite-gain $L_q$-stabilizability.} Given a system of the form \eqref{system}. We say that the control design mechanism $\pi_{\cdot}$ is {\it finite-gain $L_q$-stabilizing} if for each $\theta=[A, B]^{\top} \in \Theta$, the feedback control $\pi_{\theta}(\cdot)$  ensures the closed-loop  system \begin{align}\nonumber
    x_{t+1}(\theta)=f(Ax_{t}(\theta)+ B(\pi_{\theta}(x_{t}(\theta)) + v_t ))+w_{t+1},
\end{align}
to be finite-gain $L_q$-stable with respect to $\{v_t\}$ and $\{w_t\}$. 

When the system \eqref{system} is linear and $\pi_{\theta}(\cdot)$ is a linear state-feedback controller, i.e., $\pi_{\theta}(x_t)=K(\theta)x_t$, the finite-gain $L_q$ stability condition is equivalent to requiring that $A+BK(\theta)$ be Schur stable, which is a standard assumption in linear adaptive control (see, e.g., \cite{g1999, liunian}). Assumption~\ref{assum_stability} extends this stabilizability requirement to the nonlinear setting.
\begin{assumption}\label{assum_stability}
The control design mechanism $\pi_{\cdot}$ is finite-gain $L_\gamma$-stabilizing, where the constant $\gamma$ is given in Assumption \ref{assum_noise}. Moreover, there exist a constant $L_1>0$ such that for all  $\theta_1, \theta_2 \in \Theta$ and $x\in \mathbb{R}^n$,
\begin{equation}\label{assum4}
\|\pi_{\theta_1}(x)-\pi_{\theta_2}(x)\|\leq L_1\|\theta_1-\theta_2\|\|x\|. 
\end{equation}
\end{assumption}
\begin{remark}\label{rmk:weaker_Assumption}
Assumption~\ref{assum_stability} requires knowing a \emph{parameter-dependent} stabilizing controller $\pi_\theta(\cdot)$ for each $\theta\in\Theta$, rather than directly knowing a controller that stabilizes the true system. The former is more realistic: since the true parameter $\theta^*$ is unknown, it is generally difficult to directly specify a stabilizing controller for the true system (which depends on $\theta^*$); however, one can often characterize how to design a stabilizing controller $\pi_\theta(\cdot)$ as a function of $\theta$, should $\theta$ be known.
To illustrate, consider the linear case where \eqref{system} reduces to $x_{t+1} = Ax_t + Bu_t + w_{t+1}$. For each stabilizable pair $\theta=[A,B]^{\top}$, the LQG control design mechanism $\pi_\theta^{\mathrm{LQG}}(x_{t}) = -R^{-1}B^\top P(\theta)x_{t}$, where $P(\theta)$ solves the discrete-time algebraic Riccati equation associated with $\theta=[A,B]^{\top}$, renders the closed-loop system finite-gain $L_\gamma$-stable. In other words, while we may not know which specific controller stabilizes the true system, we do know \emph{how} to construct one given any stabilizable $\theta$.
\end{remark}
Assumption \ref{assum_stability} implies that the reference feedback law $u_t^{*}=\pi_{\theta^{*}}(x_{t})$ stabilizes the true system \eqref{system}. However, when $\theta^{*}$ is unknown, it is not obvious whether the designed adaptive controller \eqref{control}, i.e., $u_t=\pi_{\hat{\theta}_{t-1}}(x_t)+v_{t}$, still stabilizes the true system.
In what follows, Subsections~\ref{section_stable} and~\ref{section_tracking} establish the stability of the closed-loop system~\eqref{sec4_eq21} under adaptive control \eqref{control} and analyze the long-run tracking performance of the resulting state trajectory, respectively.
\subsection{Adaptive stabilization}\label{section_stable}
In this subsection, building on the asymptotic analysis of parameter estimation in Section~\ref{sec3}, we establish the global stability of the closed-loop system \eqref{sec4_eq21} and show that the parameter estimates $\hat{\theta}_t$ converge to the true parameter $\theta^{*}$ along the resulting closed-loop trajectory.

We first establish the following global stability result for the adaptive controller \eqref{control}.
\begin{theorem}\label{thm3}
Under Assumption $\ref{assum_nonlinear}$-$\ref{assum_stability}$, the system
\eqref{system} with the adaptive control $(\ref{control})$ is globally stable in the sense that, for any $x_0\in \mathbb R^n$ and $\hat\theta_0\in \mathbb R^{(m+n)\times n}$,
\begin{equation}
\limsup_{t\rightarrow \infty}\frac{1}{t}\sum_{\tau=0}^{t}\|x_{\tau}\|^{\gamma}<\infty,\, a.s.,
\end{equation}
where $\{x_t\}$ is the trajectory of the closed-loop system under adaptive control \eqref{control}, $\gamma$ is given in Assumption \ref{assum_noise}.
\end{theorem}
We note that the result of Theorem~2 continues to hold even without a probing signal, i.e., when $v_t \equiv 0$ in \eqref{control}. In addition to this stability guarantee, the next theorem shows that, if the probing signal induces sufficient excitation, the  estimates $\hat \theta_t$ converge to the ground truth $\theta^{*}$.
\begin{theorem}\label{thm2}
Under Assumptions~\ref{assum_nonlinear}–\ref{assum_stability}, if $b\in [0, \frac{(\gamma-2)^{2}}{4\gamma(\gamma+2)})$, then the parameter estimates $\hat{\theta}_{t}$ generated by Algorithm 1 along the closed-loop trajectory induced by \eqref{sec4_eq21} are strongly consistent with the following convergence rate:
\begin{equation}\label{eee40}
\|\hat{\theta}_{t}-\theta^{*}\|^{2}=O\!\left(\frac{(\log t)^{1+\delta}}{t^{1-2b-\eta}}\right),\quad \text{a.s.,}
\end{equation}
where $\eta$ is any constant satisfying $\eta \in \left(\frac{8b}{\gamma-2}, \frac{2(\gamma-2)}{\gamma(\gamma+2)}\right)$.
\end{theorem}
Proofs of Theorems \ref{thm3} and \ref{thm2} are provided in Appendix~\ref{ap3}.

Theorems~\ref{thm3} and~\ref{thm2} show that the adaptive controller  \eqref{control}, together with the estimator $\hat \theta_t$, can simultaneously stabilize the system and ensure convergence of the parameter estimates. 
Moreover, it is natural to ask whether, the resulting closed-loop trajectory can achieve a performance level comparable to that of a reference trajectory generated by 
the reference controller $u_t^{*} = \pi_{\theta^{*}}(x_t)$ which presumes knowledge of the true parameter $\theta^{*}$. This question is addressed in the next subsection.
\subsection{Adaptive tracking }\label{section_tracking}
In this subsection, we establish an adaptive tracking result under the adaptive control \eqref{control}. The performance is measured by the following average tracking error:
\begin{equation}\label{eg}
J_{t}=\frac{1}{t}\sum_{\tau=0}^{t}\left(\left\|x_{\tau}-x_{\tau}^{*}\|^{2}+ \|u_{\tau}- u_{\tau}^{*}\right\|^{2}\right).
    \end{equation}
In the sequel, we show that, if the reference controller $\pi_{\theta^*}(\cdot)$ exponentially stabilizes any equilibrium of \eqref{system}, then the adaptive controller \eqref{control}, together with the estimator $\hat{\theta}_t$ given in Algorithm~1, guarantees that the tracking error $J_t$ in \eqref{eg} vanishes asymptotically.

To this end, we introduce the following assumption on the reference trajectory $\{x_t^{*}\}_{t\geq 0}$, which requires that, under the true parameter $\theta^{*}$, the closed-loop behavior induced by $\pi_{\theta^*}(\cdot)$ is memoryless with respect to the initial condition.
\begin{assumption}\label{assum_optimality}
There exists constants $M_{0}>0$ and $\rho_{0}\in (0,1)$ such that, 
for any noise sequence $\{w_{t}\}\subset \mathbb{R}^{n}$ satisfying Assumption \ref{assum_noise} and any initial states $x_0, x_0'\in \mathbb{R}^{n}$, 
\begin{equation}\label{reference}
\|x_{t}^{*}(x_0)-x_{t}^{*}(x_0')\|\leq M_{0}\rho^{t}_{0}\|x_0-x_0'\|,\;\text{a.s.,}
\end{equation}
where $x_{t}^{*}(x_0)$ and  $x_{t}^{*}(x_0')$ are trajectories of system $\eqref{sec4_trajectory}$  
under the same noise sequence $\{w_{t}\}$
with initial conditions $x_{0}$ and $x_0'$, respectively. 
\end{assumption}
\begin{remark}
 Assumption \ref{assum_optimality} requires that the reference trajectory is exponentially contractive with respect to the initial condition. This condition ensures that transient estimation errors at the early stage do not affect the long-run average performance.  For a linear stochastic system, i.e., $f(x)=x$, if the reference trajectory is generated under the feedback control $u_{t}^{*}=K^{*}x_{t}$ such that $A^{*}+B^{*}K^{*}$ is Schur stable, then there exist constants 
$M_0>0$ and $\rho_{0}\in (0,1)$ such that
$\|(A^{*}+B^{*}K^{*})^{t}\|\leq M_0 \rho_0^{t}$ for all $t\geq 0$.
In this case, \eqref{reference}  follows immediately. 
\end{remark}

Now, we are ready to give the convergence result of the tracking error $J_t$ in the next theorem.
\begin{theorem}\label{thm4}
Suppose that Assumptions \ref{assum_nonlinear}-\ref{assum_optimality} hold, 
then for any given $b\in \left[0, \frac{(\gamma-2)^{2}}{4\gamma(\gamma+2)}\right)$ and $x_{0}\in \mathbb{R}^n$, the closed-loop trajectory under adaptive control $(\ref{control})$ 
guarantees that the average tracking error $J_t$ defined in \eqref{eg} vanishes 
with the following convergence rate:
\begin{equation}\label{eq:regret_converge_rate}
J_{t}=O\left(t^{-2b}+t^{\frac{(2b+\eta-1)(\gamma-2)}{\gamma}}(\log t)^{\frac{(1+\delta)(\gamma-2)}{\gamma}}\right),\; \text{a.s.,}
\end{equation}
where 
$\delta>0$ is given in Algorithm 1, $\gamma>2$ is given in Assumption \ref{assum_noise}, $\eta$ is any constant satisfying $\eta \in \left(\frac{8b}{\gamma-2}, \frac{2(\gamma-2)}{\gamma(\gamma+2)}\right)$. 
\end{theorem}
The proof of Theorem \ref{thm4} is provided in Appendix \ref{ap5}.

As indicated by Theorem \ref{thm4}, the performance gap between the adaptive trajectory and the reference trajectory consists of two sources: one is induced by the probing signal $v_t$ injected through the control law \eqref{control}, and the other is the estimation error $\|\hat \theta_t-\theta^*\|$ at each time step. The decay rates of these two components correspond, respectively, to the two terms in \eqref{eq:regret_converge_rate}. 

Theorem~\ref{thm4} focuses on the classical mean-square cost. In some scenarios, however, one may be interested in a more general nonlinear stage cost $c:\mathbb{R}^{n+m}\to \mathbb{R}$ and the associated trajectory tracking metric $\frac{1}{t}\sum_{\tau=0}^{t}|c(x_{t}, u_{t})-c(x_{t}^{*}, u_{t}^{*})|^2.$ In this case, the corresponding convergence result is presented in the following corollary.
\begin{corollaryx}\label{co1}
Under the conditions of Theorem~\ref{thm4}, with \eqref{reference} replaced by  
\begin{equation}\label{cost_optimality}
\begin{aligned}
&|c(x_{t}^{*}(x_{0}), u_{t}^{*}(x_{0}))-c(x_{t}^{*}(x_{0}'), u_{t}^{*}(x_{0}'))|\\
\leq &M_{0}\rho^{t}_{0}\|x_{0}-x_{0}'\|,\;\text{a.s.},
\end{aligned}
\end{equation}
and assuming that $c(\cdot)$ is Lipschitz continuous, we have that for any $b\in \left[0, \frac{(\gamma-2)^{2}}{4\gamma(\gamma+2)}\right)$ and any $x_{0}\in \mathbb{R}^n$, 
\begin{equation}\nonumber
\begin{aligned}
&\frac{1}{t}\sum_{\tau=0}^{t}|c(x_{t}, u_{t})-c(x_{t}^{*}, u_{t}^{*})|^2\\
=&O\left(t^{-2b}+t^{\frac{(2b+\eta-1)(\gamma-2)}{\gamma}}(\log t)^{\frac{(1+\delta)(\gamma-2)}{\gamma}}\right),\; \text{a.s.,}
\end{aligned}
\end{equation}
where $\delta, \gamma$, and $\eta$ are as in Theorem~\ref{thm4}.
\end{corollaryx}   
\section{Numerical Simulation}\label{sec5}
We present two application scenarios together with numerical simulations in Sections \ref{sec51} and \ref{sec52}, respectively, to verify the proposed methods.
\subsection{Pining control in opinion dynamics with a budget}\label{sec51} 
\vspace{-2mm}
Opinion dynamics typically involves nonlinearities in the form of ``S-shaped” functions, such as the hyperbolic tangent and sigmoid (see, e.g., \cite{ancona, fontan2025collective}). Here, we consider a nonlinear discrete-time opinion dynamics 
\vspace{-2mm}
\begin{equation}\label{e79}
 x_{t+1}=a\tanh\left((\alpha^* I+\tilde{A}^*)x_{t}+B^{*}u_{t}\right)+w_{t+1},
\end{equation}
 where $x_{t}=[x_{1,t}, \cdots, x_{n,t}]\in \mathbb{R}^n$ denotes the opinion state of $n$ agents at time $t$. The constant $a>0$ is the maximum sensitivity of agents,
 $\alpha^*>0$ determines the extent to which each agent reinforces its own opinion, 
and  $\tilde{A}^*\in\mathbb{R}^{n\times n}$ is the adjacency matrix of an interaction graph $\mathcal{G}$ without self-loops. The noise sequence 
$\{w_{t}\}$ is i.i.d. with $\sup_{t}\|w_{t}\|<\bar{w}$ and $\mathbb{E}[w_{t}w_{t}^{\top}]>\sigma I_n$ for some constants $\bar{w}>0$ and $\sigma>0$.
One way to intervening in the asymptotic behavior of the above opinion dynamics is through \textit{pinning control} $u_t$, whereby a selected subset of nodes is influenced by an external leader whose opinion remains constant over time.
Given a budget $M\in\mathbb{Z}^+$, let $D\subset \{1, 2,\cdots, n\}$ with $|\mathcal D|=M$
denote the set of pinned nodes. The pinning vector is defined as $B^{*}=[b_1, b_2,\cdots, b_n]^{\top}\in\mathbb{R}^{n}$, where $b_{i}=I(i\in D)$. Assume $\alpha^{*}a>1$, so that the decoupled single-agent dynamics $x_{t+1}=a\tanh(\alpha^{*}x_t)$ admits two nonzero equilibria. The desired leader opinion $x_L\in\mathbb{R}$ is chosen as one of these equilibria.

Define the system parameter as $\theta^{*}=[\alpha^{*} I_n+\tilde{A}^{*}]^{\top}$. Given $x_L$ and a parameter $\theta \in \mathbb{R}^{2n\times n}$, we can employ the pinning control law
\begin{equation}\label{eq:controller:example_1}
u_{t}=\kappa(\theta) x_L,
\end{equation}
where $\kappa(\theta)>0$  is the control gain which depends on the graph topology $\tilde{A}^*$ and reinforcement factor $\alpha^*$. Under this pinning control design, the number of nodes whose opinions eventually agree in sign with the leader's opinion is maximized under certain conditions (see \cite[Section~IV-A]{ancona} for details on the construction of $\kappa(\theta)$ and the corresponding results). Here, ``sign'' refers to the binary-opinion setting.

When the graph topology and reinforcement factor, i.e., the parameter $\theta^{*}$, is unknown, the following numerical example illustrates that our proposed method can simultaneously learn $\theta^{*}$ and synthesize an adaptive control $\hat u_{t}$, such that the resulting closed-loop opinion trajectory converges to that generated by the oracle control $u^{*}_{t}$. Here, $u^{*}_{t}$ is the controller by setting $\theta=\theta^*$ in \eqref{eq:controller:example_1}.
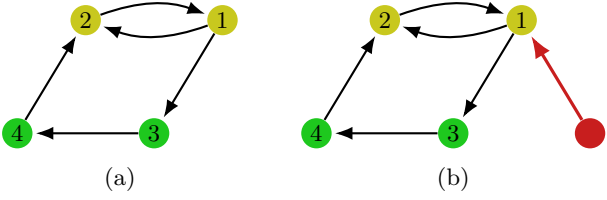
\begin{figure}[tbp]
\centering

\begin{minipage}[t]{0.48\linewidth}
\centering
\begin{tikzpicture}[
  >=Latex,
  node1/.style={circle,fill=myyellow,minimum size=4mm,inner sep=0pt,font=\small},
  node2/.style={circle,fill=mygreen,minimum size=4mm,inner sep=0pt,font=\small},
  every edge/.style={-Latex,thick},
  shorten >=1pt
]
\node[node1] (n2) at (0,0) {2};
\node[node1] (n1) at (1.8,0) {1};
\node[node2] (n3) at (0.9,-1.5) {3};
\node[node2] (n4) at (-0.9,-1.5) {4};

\draw[->,thick] (n1) to[bend left=20] (n2);
\draw[->,thick] (n2) to[bend left=20] (n1);
\draw[->,thick] (n1) to (n3);
\draw[->,thick] (n3) -- (n4);
\draw[->,thick] (n4) to (n2);
\end{tikzpicture}

\vspace{1mm}
{\small (a)}
\end{minipage}
\hfill
\begin{minipage}[t]{0.48\linewidth}
\centering
\begin{tikzpicture}[
  >=Latex,
  node1/.style={circle,fill=myyellow,minimum size=4mm,inner sep=0pt,font=\small},
  node2/.style={circle,fill=mygreen,minimum size=4mm,inner sep=0pt,font=\small},
  pin/.style ={circle,fill=myred,minimum size=4mm,inner sep=0pt},
  every edge/.style={-Latex,thick},
  shorten >=1pt
]
\node[node1] (n2) at (0,0) {2};
\node[node1] (n1) at (1.8,0) {1};
\node[node2] (n3) at (0.9,-1.5) {3};
\node[node2] (n4) at (-0.9,-1.5) {4};
\node[pin]  (p)  at (2.7,-1.5) {};

\draw[->,thick] (n1) to[bend left=20] (n2);
\draw[->,thick] (n2) to[bend left=20] (n1);
\draw[->,thick] (n1) to (n3);
\draw[->,thick] (n3) -- (n4);
\draw[->,thick] (n4) to (n2);

\draw[myred,->,very thick] (p) to (n1);
\end{tikzpicture}

\vspace{1mm}
{\small (b)}
\end{minipage}

\caption{Topology of the four agents. Yellow and green represent different initial opinions of the agents; the red node denotes the leader (external controller), and the red arrow indicates the pinning input applied to agent 1.}
\label{fig:1}
\end{figure}
\begin{example}
Consider the nonlinear opinion dynamics \eqref{e79} on a four-agent network, as illustrated in Fig.~\ref{fig:1}(a). The noise sequence $\{w_t\}_{t\geq 0}$ is i.i.d. and uniformly distributed over $[-0.1, 0.1]^4$. Let $ a=2$, $\alpha^{*} = 0.7$, and let the corresponding adjacency matrix $\tilde{A}^{*}$ be
\[
\tilde{A}^{*}=\begin{bmatrix} 
0 & 0.4 & 0 & 0\\ 
0.5 & 0 & 0 & 0.1\\
0.3 & 0 & 0 & 0\\
0 & 0 & 0.5 & 0\\
 \end{bmatrix}.
\] 
The desired leader opinion in this example is calculated as $x_L=1.63$. Additionally,  we set $M=1$ and $B^{*}=[1,0, 0, 0]^{\top}$, meaning that agent $1$ is the pinned node, as shown in Fig.
\ref{fig:1}(b). 
Let $\Theta=\{\theta \in \mathbb{R}^{(n+m)\times n}: \|\theta\|\leq 5 \}$ and initialize $\hat{\theta}_0=0$. The sequence ${\hat{\theta}_t}$ is generated using Algorithm~1, and the adaptive controller is implemented as $\hat u_{t}=\kappa(\hat \theta_t) x_L+t^{-\frac{1}{8}}\epsilon_t.$,
where $\kappa(\cdot)$ is designed according to \cite[Section
IV.A]{ancona}, and the probing signal
 $\{\epsilon_t\}_{t\geq 0}$ is i.i.d. and uniformly distributed on $[-1, 1]$.  
By Theorem \ref{thm2} and Corollary \ref{co1}, the parameter estimates are almost surely consistent, and the trajectory error associated with the sign function satisfies that
\begin{equation}\label{opinion_regret}
\limsup_{t\to \infty}\frac{1}{t}\sum_{\tau=0}^{t}\sum_{i=1}^{4}\|sgn(x_{i,\tau})\!-\!sgn(x_{i,\tau}^{*})\|\!=\!0,\,a.s.,
\end{equation}
where $x_{i,t}$ and $x_{i,t}^{*}$ denote the opinion trajectory under adaptive control $\hat u_t$ and oracle control $u^*_t$, respectively.
\vspace{-1mm}

In this setup, the parameter estimation error converges to zero, as shown in Fig.\ref{fig_1}(a).
The regret trajectory defined in \eqref{opinion_regret} is also plotted in Fig.~\ref{fig_1}(b). As can be seen, under the adaptive control $\hat{u}_t$, the time-averaged opinion of each agent converges to that under the oracle control $u^{*}_t$ associated with the true parameter $\theta^{*}$.
\begin{figure}[!t]
\centering
\includegraphics[width=3.5in]{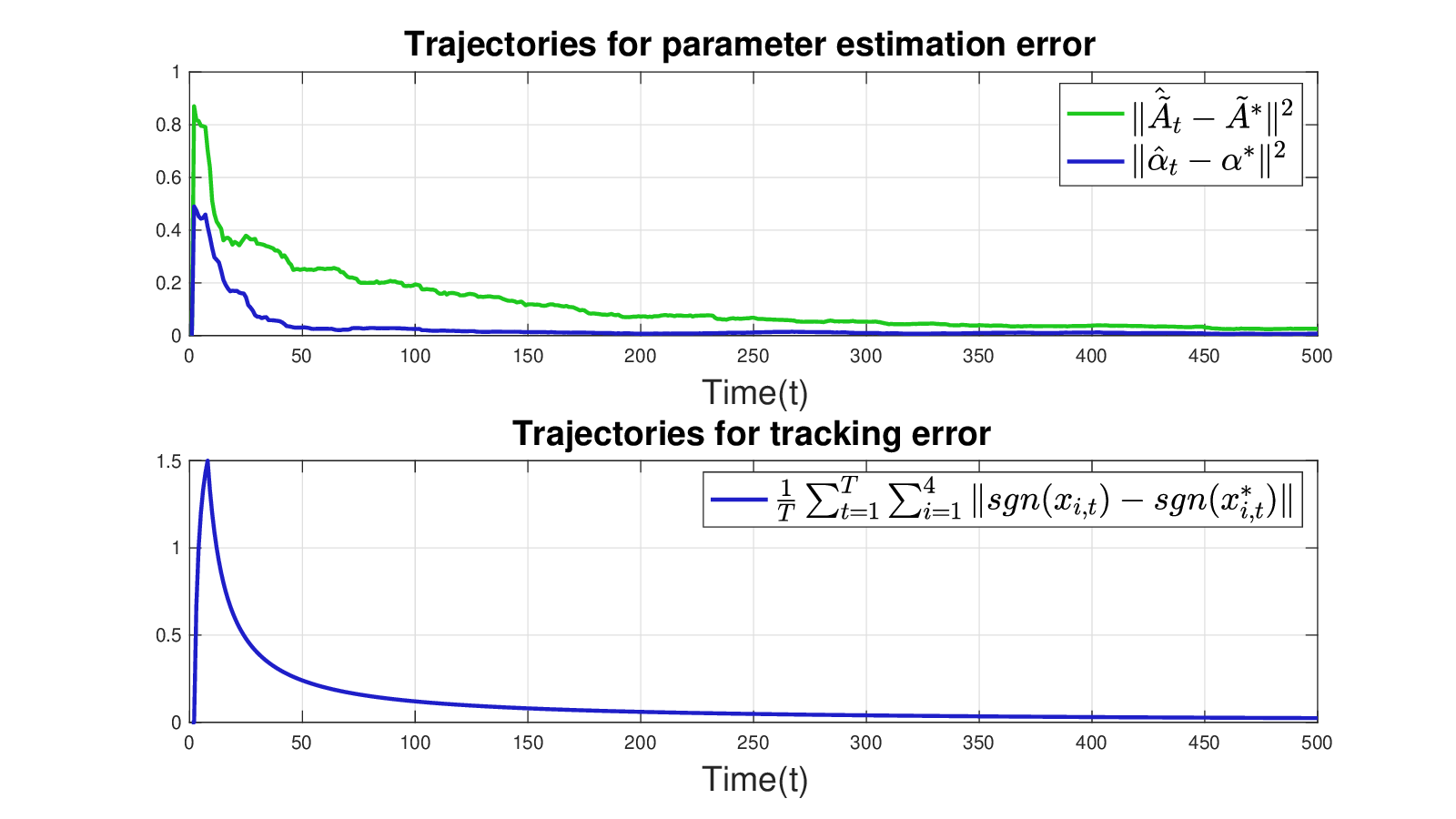}
\caption{Trajectories for the parameter estimation error and average performance tracking error.}
\label{fig_1}
\end{figure}
\end{example}

\subsection{Epidemic models with multi-community}\label{sec52} 
Consider a susceptible-infective (SI) model for the epidemic evolution across $n$ networks of individuals, 
\begin{align}
    \begin{cases}
        x_{i,t+1}\! =\!\!\! \left[ (1+\beta s_{i,t}) x_{i,t} +\sum\limits_{ j\not=i}\alpha s_{i,t} x_{j,t}-\bar{u}_{i,t}+\eta_{i,t+1}\right]_+ \\
s_{i,t} = N-x_{i,t}, \label{eq:SI_model_1}
    \end{cases}
\end{align}
where state $x_{i,t}$ and $s_{i,t}$ represent the levels of infective and susceptible in network $i$ at time $t$, respectively. The constant $N$ is the number of individuals in each network, and $\alpha, \beta \in \mathbb R$ are the intra-network and cross-network infection rates. For each network $i\in[n]$, the action $\bar{u}_{i,t}$ represents a controlled removal rate, and the noise sequence $\{\eta_{i, t}\}$ is i.i.d. Gaussian with distribution $\mathcal{N}(0,\sigma^{2})$. Additionally, the operator $[\cdot]_{+}: \mathbb{R}\to \mathbb{R}$ is defined by $[z]_{+}=z$ if $0\leq z \leq N$, $[z]_{+}=N$ if $z>N$, and $[z]_{+}=0$ otherwise.

Without loss of generality, we consider the case $n=2$ with a complete graph.\footnote{A general connection topology with $n>2$ can be handled in a similar manner.}, and define $x_{t}=[x_{1,t}, x_{2,t}]^{\top},$ $u_{t}=[x_{1,t}^{2}, x_{2,t}^{2}, x_{1,t}x_{2,t}, \bar{u}_{1,t}, \bar{u}_{2,t}]^{\top}$, $\eta_{t}=[\eta_{1,t}, \eta_{2,t}]^{\top}.$ Let
\[
A^{*}=\left[\begin{matrix} 
\beta N & \alpha N \\ 
\alpha N & \beta N \\
 \end{matrix}\right], \;\; 
 B^{*}=\left[\begin{matrix} 
-\beta & 0 & -\alpha & -1 & 0\\ 
0 & -\beta & -\alpha & 0 & -1\\
 \end{matrix}\right].
 \]
Then the dynamics \eqref{eq:SI_model_1} can be rewritten as $x_{t+1}=[A^{*}x_{t}+B^{*}u_{t}+\eta_{t+1}]_{+}$, where the operator $[\cdot]_{+}$ is applied elementwise. Accordingly, we define the system parameter as
$\theta^{*}=[A^{*}, B^{*}]^{\top}$, which encodes the graph topology, the infection rates, and the community sizes.

By an argument similar to that in Example~\ref{ex1}, noise entering inside the nonlinear map can be recast as an equivalent additive external noise. Specifically, the system \eqref{eq:SI_model_1} can be rewritten as
\begin{equation}\label{positive}
\begin{aligned}
x_{t+1}=h(A^{*}x_{t}+B^{*}u_{t})+w_{t+1},
\end{aligned}
\end{equation}
where $h:\mathbb{R}^n\to\mathbb{R}^n$ is defined componentwise by $h(z)=[h_1(z_1),\ldots,h_n(z_n)]^{\top}$, with 
$h_i(z_i)=\mathbb{E}\!\left[[z_i+\eta_{i,t+1}]_{+}\right]=N-z_iG(-z_i)-(N-z_i)G(N-z_i)+\sigma^{2}[g(-z_i)-g(N-z_i)]$. Here, $G(\cdot)$ and $g(\cdot)$ denote, respectively, the cumulative distribution function and the probability density function of the Gaussian distribution  $\mathcal{N}(0, \sigma^2)$.
The induced external noise 
$w_{t+1}=x_{t+1}-h(A^{*}x_t+B^{*}u_t)
$ is a martingale difference sequence. It is straightforward to verify that the map $h$ and $\{w_t\}$ satisfy Assumptions \ref{assum_nonlinear} and \ref{assum_noise}.

When the parameter $\theta^{*}$ is known,  a stabilizing feedback control can be designed in the form  $\bar{u}_{t}^{*}=\bigl(R + P^{*}\bigr)^{-1} P^{*} A^{*}x_t^{*}$, where $P^{*}\in \mathbb R^{2\times 2}$ is the solution to the discrete-time Riccati equation $P^{*} = A^{*\top} P^{*} A^{*} - A^{*\top} P^{*} (R + P^{*})^{-1} P^{*} A^{*} + Q,$ and $Q>0, R>0$ are given weighting matrices defining the quadratic objective to be minimized (see \cite{bloem} for details). The following example demonstrates that the proposed adaptive controller attains the same asymptotic performance without prior knowledge of $\theta^*$.
\begin{example}
In \eqref{positive},
 we set the parameters $\alpha^{*}=0.15$, $\beta^{*}=0.3$ and $N=10$. Define the feasible set $\Theta=\left\{[A, B]^{\top}: \|A\|\leq 15, \|B\|\leq 5, A\in  \mathbb{R}^{2\times 2}, B\in \mathbb{R}^{2\times 5} \right\}$, and let
$\Theta_{\epsilon}=\left\{\theta : 2\theta\in \Theta \right\}$. The parameter estimates  $\hat{\theta}_t$ are generated by Algorithm 1. Following the adaptive control structure in \eqref{control}, we apply the adaptive control  $u_{t}=[x_{1,t}^{2}, x_{2,t}^{2}, x_{1,t}x_{2,t}, \bar{u}_{1,t}, \bar{u}_{2,t}]^{\top}$ with $[\bar{u}_{1,t}, \bar{u}_{2,t}]^\top= \bigl(I + \hat{P}_{t}\bigr)^{-1} \hat{P}_{t} \hat{A}_{t}x_t+t^{-\frac{1}{8}}\epsilon_t.$ Here $\epsilon_t=[\epsilon_{1,t}, \epsilon_{2,t}]^{\top}$, with $\{\epsilon_{i,t}\}$ i.i.d. and uniformly distributed on $[-1,1]$. The matrix $\hat{P}_{t}$ is the unique solution to the Riccati equation $\hat{P}_{t} = \hat{A}_{t}^{\top} \hat{P}_{t} \hat{A}_{t} - \hat{A}_{t}^{\top}\hat{P}_{t} (I + \hat{P}_{t})^{-1} \hat{P}_{t}\hat{A}_{t} + I,$ where we choose $Q=R=I_{2}$. 

In Fig.\ref{fig_2}, we plot the trajectories of the parameter estimation error under different noise standard deviations $\sigma=1, 5, 10$. In all three cases, the estimation error decreases as $t$ increases, which is consistent with Theorem \ref{thm2}. This is because larger noise inject stronger excitation during the estimation process. Moreover, the trajectory of tracking error is also shown in Fig.~\ref{fig_2}, which shows that both the state trajectory and the control signal track the reference ones asymptotically. 
\begin{figure}[!t]
\centering
\includegraphics[width=3.4in]{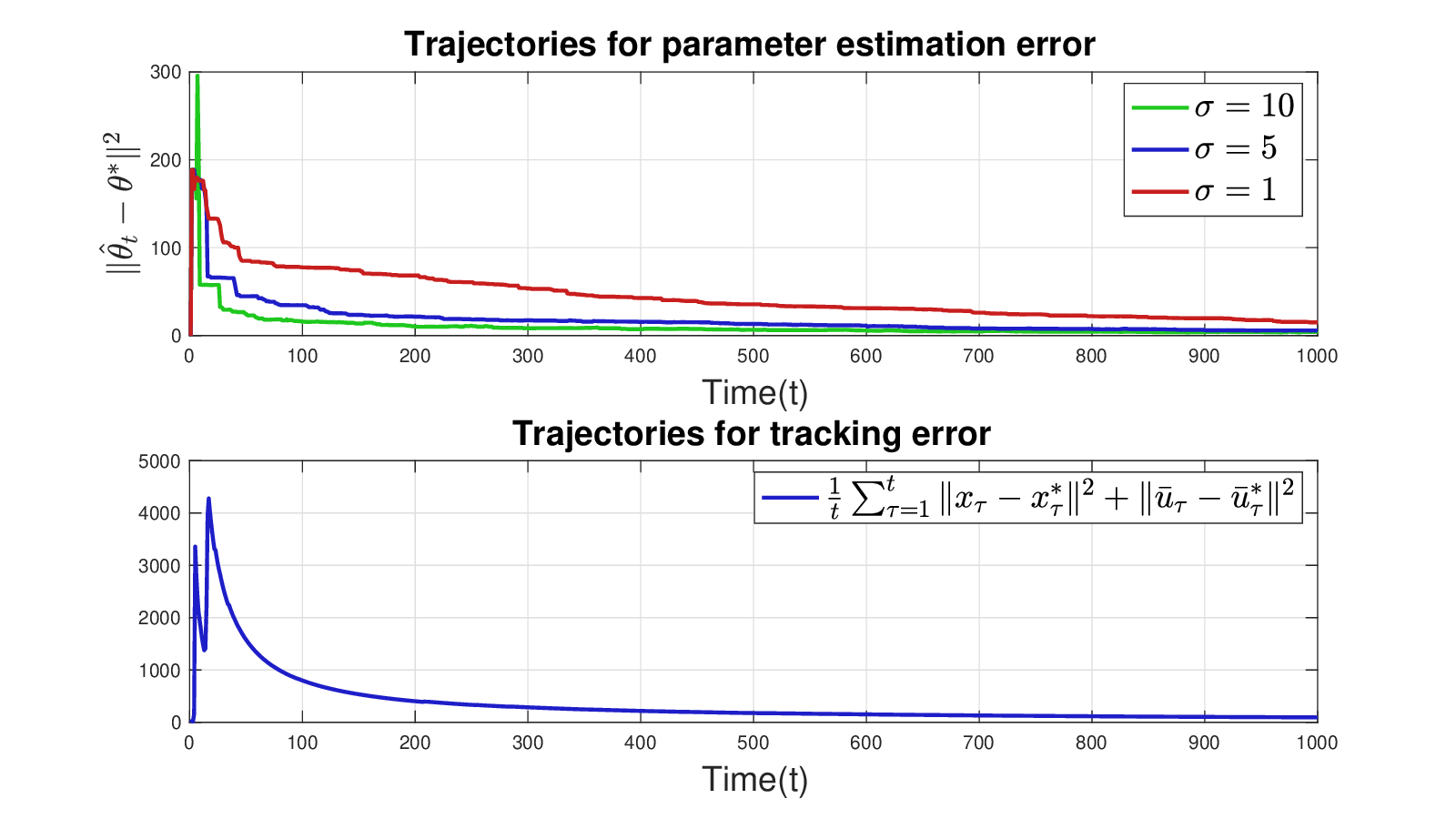}
\caption{Trajectories of the parameter estimation error for different noise standard deviations $\sigma$ and the average tracking performance error. 
}
\label{fig_2}
\end{figure}
\end{example}
\section{Conclusions}\label{sec6}	
This paper investigated the online identification and adaptive control for nonlinear stochastic systems in which the state update is given
by a nonlinear function of linear dynamics with additive random noise. A new online identification algorithm is proposed to estimate the unknown parameters, and is shown that will converge to the true parameter almost surely under
a general non-PE condition on the trajectory data. These identification results further allow us to establish
global stability and asymptotic tracking optimality for the resulting adaptive closed-loop system.
We note that the analysis can be extended to other settings, including the case where Gaussian noise enters inside the nonlinear map $f(\cdot)$, provided that $f(\cdot)$ is bounded and monotone, as illustrated in Example \ref{ex1}. However, many interesting problems remain open for future investigation, such as the design and analysis of adaptive control with general observation matrices, and extensions to more complex nonlinear systems including multi-layer networks.
\appendix
\setlength{\abovedisplayskip}{5pt}   
\setlength{\belowdisplayskip}{5pt}   
\setlength{\abovedisplayshortskip}{1.5pt}
\setlength{\belowdisplayshortskip}{1.5pt}
    \setlength{\jot}{3pt}  
\setlength{\parindent}{0.85em}  
\setlength{\parskip}{0.5\baselineskip} 
\section{Proof of Theorem \ref{thm1}}\label{ap2}
For notational simplicity, define
\begin{equation}
\psi_{t}=f(\theta^{*\top}\phi_{t})-f(\hat{\theta}_{t}^{\top}\phi_{t}), \quad t\geq 0.
\end{equation}
Since both $\phi_{t}$ and $\hat{\theta}_{t}$ are $\mathcal{F}_{t}-$measurable, it follows that $\psi_{t}$ is also  $\mathcal{F}_{t}-$measurable. As a key step toward Theorem~\ref{thm1}, we first establish the following lemma.
\begin{lemma}\label{lem6}
	Under the conditions of Theorem~\ref{thm1}, the parameter estimates produced by Algorithm~1 have the following property as $t\to \infty$:	\begin{equation}\nonumber
		\begin{aligned}
\operatorname{tr}[(\theta^{*}-\hat{\theta}_{t})^{\top}P_{t}^{-1}(\theta^{*}-\hat{\theta}_{t})]+\sum_{\tau=0}^{t}a_{\tau}\|\psi_{\tau}\|^{2}= O\left(1\right),\;\; \text{a.s.,}
		\end{aligned}
	\end{equation}
    where $\{a_{\tau}\}_{\tau\ge 0}$ is defined in \eqref{sec3_alg11}.
\end{lemma}
The proof of Lemma~\ref{lem6} is provided in Appendix~\ref{ap7}.
\begin{pf}[Proof of Theorem \ref{thm1}.] We first prove property (i). Taking trace on both sides of \eqref{eq6} and summing
up, we have
\begin{equation}\label{a.5}
\sum_{\tau=0}^{\infty}a_{\tau}d_{\tau}^{2}\phi_{\tau}^{\top}P_{\tau}^{2}\phi_{\tau}\leq \sum_{\tau=1}^{t}[tr(P_\tau)-tr(P_{\tau+1})]<\infty.
\end{equation}
Besides, according to Cauchy-Schwarz inequality, 
\begin{equation}\nonumber
\begin{aligned}
\sum_{\tau=0}^{\infty}\left\|a_{\tau}d_{\tau}P_{t}\phi_{\tau}\psi_{\tau}^{\top}\right\|
\leq \sum_{\tau=0}^{\infty}a_{\tau}d_{\tau}^{2}\phi_{\tau}^{\top}P_{t}^{2}\phi_{\tau}+\sum_{\tau=0}^{\infty}a_{\tau}\left\|\psi_{\tau}\right\|^{2}.\\
\end{aligned}
\end{equation}
Therefore, by \eqref{a.5} and Lemma \ref{lem6}, we have
\begin{equation}\label{first_term}
\sum_{\tau=0}^{\infty}\left\|a_{\tau}d_{\tau}P_{\tau}\phi_{\tau}\psi_{\tau}^{\top}\right\|<\infty,\;\text{a.s.}   
\end{equation}
As for the noise term, note that $a_t \le 1$ by \eqref{sec3_alg11}. Using \eqref{a.5} again and by Assumption~\ref{assum_noise}, we know that
\[\sum_{\tau=0}^{\infty}\mathbb{E}[\|a_{\tau}d_{\tau}P_{\tau}\phi_{\tau}w_{t+1}\|^{2}\mid \mathcal{F}_t]<\infty,\;\text{a.s.}\]
So by Chow's martingale convergence theorem (See Theorem 2.7 in \cite{cg1991}), we know that $\left\|\sum_{\tau=0}^{\infty}a_{\tau}d_{\tau}P_{\tau}\phi_{\tau}w_{\tau+1}\right\|<\infty$ almost surely. Hence, combining with \eqref{first_term}, one can obtain
\begin{equation}\label{final}
\sum_{\tau=0}^{\infty}\left\|a_{\tau}d_{\tau}P_{\tau}\phi_{\tau}\left(\psi_{\tau}+w_{\tau+1}\right)^{\top}\right\|<\infty,\;\text{a.s.}
\end{equation}

We now prove the claim that there exists a finite stopping time $t_1>0$ such that for all $t \ge t_1$, the projection is no longer active i.e., $\hat{\theta}_{t+1}=\hat{\theta}_{t}+\mu_{t}^{-1}d_{t}P_{t+1}\phi_{t}[x_{t+1}-f(\hat{\theta}_{t}^{\top}\phi_{t})]^{\top}, \forall t>t_1$.
Indeed, from \eqref{final}, there exists a finite stopping time $t_0 \ge 0$ such that 
\begin{equation}\label{b.6}
\sum_{\tau=t_0}^{\infty}\left\|a_{\tau}d_{\tau}P_{\tau}\phi_{\tau}\left(\psi_{\tau}+w_{\tau+1}\right)^{\top}\right\|<\frac{\epsilon}{2},\;\text{a.s.},
\end{equation} 
where $\epsilon$ is given in Assumption \ref{assum_parameter}. For each sample point $\omega$, if $
\hat{\theta}_{t}+a_{t}d_{t}P_{t}\phi_{t}\left(\psi_{t}+w_{t+1}\right)^{\top}\in  \Theta, t\geq t_0$, then the claim holds with $t_1=t_0$.
Otherwise, there exists a finite time $t_0'\geq t_0$ such that
$
\hat{\theta}_{t_0'}+a_{t_0'}d_{t_0'}P_{t_0'}\phi_{t_0'}\left(\psi_{t_0'}+w_{t_0'+1}\right)^{\top}\notin  \Theta.
$ By the definition of the projection, we have 
$\hat{\theta}_{t_0'+1}\in \Theta_{\epsilon}.$ Then, for all $t\geq t_0'+1$, \eqref{b.6} ensures that $\hat{\theta}_{t_0'}+\sum_{\tau=t_0'}^{t}a_{\tau}d_{\tau}P_{\tau}\phi_{\tau}\left(\psi_{\tau}+w_{\tau+1}\right)^{\top}\in  \Theta$. Thus, $
\hat{\theta}_{t}+a_{t}d_{t}P_{t}\phi_{t}\left(\psi_{t}+w_{t+1}\right)^{\top}\in  \Theta, \forall t\geq t_{0}'+1,
$ and the claim again holds with $t_1=t_{0}'+1.$ Therefore, the claim holds.
By the claim, the estimate satisfies $\hat{\theta}_{t+1}=\hat{\theta}_{t_1}+\sum_{\tau=t_1}^{t}a_{\tau}d_{\tau}P_{\tau}\phi_{\tau}[x_{\tau+1}-f(\hat{\theta}_{\tau}^{\top}\phi_{\tau})]^{\top}=\hat{\theta}_{t_1}+\sum_{\tau=t_1}^{t}a_{\tau}d_{\tau}P_{\tau}\phi_{\tau}\left(\psi_{\tau}+w_{\tau+1}\right)^{\top}, \forall t\geq t_1.$
From this and \eqref{final}, it follows that property (i) holds.

We now prove property (ii). To begin, we show that
\begin{equation}\label{b9}
\underline{d}:=\inf_{t\geq 0}d_{t}>0,\;\text{and} \;\sup_{t\geq 0}\overline{g}_{t}<\infty,\;\text{a.s.}
\end{equation}
If $\sup_{x\in\mathbb{R}^{n}}|f(x)|<\infty$, then the state trajectory is uniformly bounded almost surely, i.e., $\sup_{t}\|x_t\|<\infty$, a.s.
Together with the condition in Theorem \ref{thm1} that $\|u_t\|\le c_1\|x_t\|+c_2$, we also have $\sup_{t}\|\phi_t\|<\infty$, a.s.
From property (i), $\sup_{t}\|\hat{\theta}_{t}\|<\infty, a.s.$ Therefore, there exists a real-valued random variable $C_1>0$ such that $\sup_{t}\left(\|\hat{\theta}_t^{\top}\phi_{t}\|+\|\theta^{*\top}\phi_{t}\|\right)\leq C_1$. Then, by the monotonicity of $\alpha(\cdot)$ and $\beta(\cdot)$ in Assumption~\ref{assum_nonlinear}, we obtain
$\inf_{t\geq 0}d_{t}\geq \frac{1}{2}\alpha(C_1)>0$, and $\sup_{t\geq 0}\overline{g}_{t}\leq \beta(C_1)<\infty$. If $\sup_{x\in \mathbb{R}^{n}}\|f(x)\|=\infty,$ then we have $d_{t}=\frac{1}{2}\alpha(\|\hat{\theta}_t^{\top}\phi_{t}\|+\max_{\theta\in \Theta}\|\theta^{\top}\phi_{t}\|)\geq \frac{1}{2}\inf_{z\in \mathbb{R}^n}\alpha(\|f(z)\|)$ and $\overline{g}_{t}=\beta(\|\hat{\theta}_t^{\top}\phi_{t}\|+\max_{\theta\in \Theta}\|\theta^{\top}\phi_{t}\|)\leq\inf_{z\in \mathbb{R}^n}\beta(\|f(z)\|)$.
By Assumption \ref{assum_nonlinear}, $\sup_{t}\overline{g}_{t}\leq \sup_{z\in \mathbb{R}^n}\beta(\|f(z)\|)<\infty$ and $\inf_{t}d_{t}\geq \frac{1}{2}\inf_{z\in \mathbb{R}^n}\alpha(\|f(z)\|)>0$, which proves \eqref{b9}.

Moreover, since $\overline{g}_t\geq d_t$ and $d_{t}\geq \underline{d}$ for all $t\geq 0$, it follows that $d_{t}\overline{g}_t^{2}\phi_{t}^{\top}P_{t}\phi_{t}\geq \underline{d}d_{t}^{2}\phi_{t}^{\top}P_{t}\phi_{t}$. Then, by the definition of $a_{t}$ and $\mu_{t}$ in \eqref{sec3_alg11} and Algorithm $1$, respectively,
\begin{equation}\label{eqb8}
\mu_{t}^{-1}\leq (1+\underline{d}^{-1})a_{t},\forall t\geq 0.
\end{equation}
Finally, combining \eqref{b9}, \eqref{eqb8} with Lemma~\ref{lem6} gives property (ii) directly.

To prove property (iii), 
 from $(\ref{eq6})$ and the Sherman–Morrison formula \citep{hager}, we have
\begin{equation}\label{b8}
P_{t+1}^{-1}=P_{t}^{-1}+\frac{d_{t}^{2}}{\mu_{t}}\phi_{t}\phi_{t}^{\top}=P_{0}^{-1}+\sum_{\tau=0}^{t}\frac{d_{\tau}^{2}}{\mu_{\tau}}\phi_{\tau}\phi_{\tau}^{\top}.
\end{equation}
Using \eqref{b9} and the fact that $\|P_{t}\|\leq \|P_{0}\|$ for all $t\ge 0$, there exists a real-valued random variable $C_2>0$ such that $|\mu_{t}|=|(\log r_t)^{1+\delta}+\frac{1}{2}d_{t}\overline{g}_{t}^{2}\phi_{t}^{\top}P_{t}\phi_{t}|\leq C_2[(\log r_t)^{1+\delta}(1+\|\phi_{t}\|^{2})],$a.s. 
Substituting this bound into \eqref{b8}, we have
\begin{equation}\label{eqb10}
\lambda_{\min}\{P_{t}^{-1}\}\geq \frac{\underline{d}^2}{C_2(\log r_{t})^{1+\delta}}\lambda_{\min}\left\{\sum_{\tau=0}^{t-1} \frac{\phi_{\tau}\phi_{\tau}^{\top}}{1+\|\phi_{\tau}\|^{2}}\right\},\text{a.s.}
\end{equation}
Moreover, we know that
\[
\|\theta^{*}-\hat{\theta}_{t}\|^{2}\leq \frac{\operatorname{tr}[(\theta^{*}-\hat{\theta}_{t})^{\top}P_{t}^{-1}(\theta^{*}-\hat{\theta}_{t})]}{\lambda_{\min}\{P_{t}^{-1}\}}. 
\]
Finally, by combining this inequality with Lemma~\ref{lem6}, \eqref{eqb10}, and \eqref{b9}, we conclude property (iii) holds.
\end{pf}
\titlespacing*{\section}{0pt}{0.5em}{0.3em}
\section{Proof of Theorem \ref{thm3}}\label{ap3}
\begin{pf}[Proof of Theorem \ref{thm3}.] For notational convenience, we define the sequences
$\{\bar r_t\}_{t\geq 0}$, $\{\bar\mu_t\}_{t\geq 0}$, and $\{h_t\}_{t\geq 0}$ as follows:
\begin{equation}\nonumber
\begin{aligned}
\bar{r}_{t}&=1+\sum_{\tau=0}^{t}\|\phi_{\tau}\|^{\gamma},\; \bar{\mu}_t=d_{t}\overline{g}_t^{2}\phi_{t}^{\top}P_{t}\phi_{t},\\
h_{t}&=\frac{\|(\theta^{*}-\hat{\theta}_{t})^{\top}\phi_{t}\|^{2}}{2\bar{\mu}_t}I\left(\bar{\mu}_t\geq (1+\log r_t)^{1+\delta}\right).
\end{aligned}
\end{equation}
From property (i) of Theorem~\ref{thm1}, we have $lim_{t\to \infty}\hat{\theta}_{t}=\bar{\theta},\;a.s.$ Thus, we obtain
\begin{equation}\label{eqq33}
\begin{aligned}
&\sum_{\tau=0}^{t}\|(\hat{\theta}_{\tau}-\bar{\theta})^{\top}\phi_{\tau}\|^{\gamma}=o\left(\sum_{\tau=0}^{t}\|\phi_{\tau}\|^{\gamma}\right)+O(1)\\
&=o(\bar{r}_{t})+O(1),\;\text{a.s.},\;t\to \infty.
\end{aligned}
\end{equation}
Next, by the definition of $\bar{\mu}_{t}$ and the bound \eqref{b9}, together with  $\|P_t\|\leq \|P_0\|, \forall t\geq 0$, we have
\begin{equation}\label{c3}
|\bar{\mu}_{t}|^{\frac{\gamma}{2}}=O\left(\|\phi_{t}\|^{\gamma}\right),\;\text{a.s.},\;t\to \infty.
\end{equation}
Moreover, property (ii) of Theorem~\ref{thm1} gives
\begin{equation}\label{c4}
\sum_{\tau=0}^{t}h_{\tau}\leq \sum_{\tau=0}^{t}\mu_{\tau}^{-1}\|(\theta^{*}-\hat{\theta}_{\tau})^{\top}\phi_{\tau}\|^{2}=O(1),\;\text{a.s.}, t\to \infty,
\end{equation}
and similarly,
\begin{equation}\label{c5}
\begin{aligned}
&\sum_{\tau=0}^{t}\frac{\|(\theta^{*}-\hat{\theta}_{\tau})^{\top}\phi_{\tau}\|^{2}}{2(1+\log r_{\tau})^{1+\delta}}I\left(\bar{\mu}_{\tau}<(1+\log r_{\tau})^{1+\delta}\right)\\
\leq&\sum_{\tau=0}^{t}\mu_{\tau}^{-1}\|(\theta^{*}-\hat{\theta}_{\tau})^{\top}\phi_{\tau}\|^{2}=O(1),\;\text{a.s.}, t\to \infty.
\end{aligned}
\end{equation}
From \eqref{c4}, we conclude that
\begin{equation}\label{eqb5}
|h_{t}|^{\gamma}=o(1),\;\text{a.s.},t\to \infty.
\end{equation}
Note that $\|(\theta^{*}-\hat{\theta}_{\tau})^{\top}\phi_{\tau}\|^{2}I\left(\bar{\mu}_{\tau}\geq (1+\log r_{\tau})^{1+\delta}\right)=2\bar{\mu}_{\tau}h_{\tau}$. Then, by \eqref{c3} and \eqref{eqb5}, we have
\begin{equation}\label{c7}
\begin{aligned}
&\sum_{\tau=0}^{t}\|(\theta^{*}-\hat{\theta}_{\tau})^{\top}\phi_{\tau}\|^{\gamma}I\left(\bar{\mu}_{\tau}\geq (1+\log r_{\tau})^{1+\delta}\right)\\
=&\sum_{\tau=0}^{t}|\bar{\mu}_{\tau}h_{\tau}|^{\frac{\gamma}{2}}=o(\bar{r}_{t}), \;\text{a.s.}, t\to \infty.
\end{aligned}
\end{equation}
Then, from \eqref{c5}, we also have
\begin{equation}\label{c8}
\begin{aligned}
&\sum_{\tau=0}^{t}\|(\theta^{*}-\hat{\theta}_{\tau})^{\top}\phi_{\tau}\|^{\gamma}I\left(\bar{\mu}_{\tau}< (1+\log r_{\tau})^{1+\delta}\right)\\
=&O((1+\log r_{t})^{\frac{(1+\delta)\gamma}{2}})=O((\log t+\log \bar{r}_{t})^{\frac{(1+\delta)\gamma}{2}}),\text{a.s.}
\end{aligned}
\end{equation}
Combining \eqref{c7} and \eqref{c8} gives 
\begin{equation}\nonumber
\sum_{\tau=0}^{t}\|(\theta^{*}-\hat{\theta}_{\tau})^{\top}\phi_{\tau}\|^{\gamma}=o(\bar{r}_t)+O((\log t)^{\frac{(1+\delta)\gamma}{2}}), \text{a.s.}
\end{equation}
Using this together with \eqref{eqq33}, we obtain
\begin{equation}\label{c9}
\sum_{\tau=0}^{t}\|(\theta^{*}-\bar{\theta})^{\top}\phi_{\tau}\|^{\gamma}=o(\bar{r}_t)+O((\log t)^{\frac{(1+\delta)\gamma}{2}}),\text{a.s.}
\end{equation}
Additionally, under Assumption~\eqref{assum4}, the controller mapping satisfies $\|\pi_{\hat{\theta}_{\tau-1}}(x_{\tau})-\pi_{\bar{\theta}}(x_{\tau})\|\leq L_1\|\hat{\theta}_{\tau-1}-\bar{\theta}\|\|x_{\tau}\|+L_2$, and therefore,
\begin{equation}\label{eqq35}
\sum_{\tau=0}^{t}\|\pi_{\hat{\theta}_{\tau-1}}(x_{\tau})-\pi_{\bar{\theta}}(x_{\tau})\|^{\gamma}=o(\bar{r}_{t})+O(1),a.s.
\end{equation}
For each $t\geq 0$, define $\bar{v}_{t}=(A^{*}-\bar{A})x_t+B^{*}\pi_{\hat{\theta}_{t-1}}(x_{t})-\bar{B}\pi_{\bar{\theta}}(x_{t})+B^{*}v_{t}.$ Then, for each $t\geq 0$,
\begin{equation}\nonumber
\begin{aligned}
\bar{v}_{t}=[\theta^{*}-\bar{\theta}]^{\top}\phi_t+\bar{B}\pi_{\hat{\theta}_{t-1}}(x_{t})-\bar{B}\pi_{\bar{\theta}}(x_{t})+\bar{B}v_{t}.
\end{aligned}
\end{equation}
Then, from Assumption~\ref{assum_noise} and \eqref{sec4eq23}, the noise and input sequences satisfy, almost surely,
$\sum_{\tau=1}^{t}\|w_{\tau}\|^{\gamma}=O(t),$ and $\sum_{\tau=0}^{t-1}\|v_{\tau}\|^{\gamma}=O(t)$.
Combining these with \eqref{c9} and \eqref{eqq35}, it follows that
\begin{equation}\label{eqc10}
\sum_{\tau=0}^{t-1}\|\bar{v}_{\tau}\|^{\gamma}+\sum_{\tau=1}^{t}\|w_{\tau}\|^{\gamma}=o(\bar{r}_{t})+O(t),\; a.s.
\end{equation}
Note that the system can be rewritten as
\begin{equation}
\begin{aligned}
x_{t+1}=&f\left(A^{*}x_t+B^{*}\pi_{\hat{\theta}_{t-1}}(x_{t})+B^{*}v_{t}\right)+w_{t+1}\\
=&f\left(\bar{A}x_t+\bar{B}\pi_{\bar{\theta}}(x_{t})+\bar{v}_t\right)+w_{t+1}.
\end{aligned}
\end{equation}
From Assumption \ref{assum_stability} and \eqref{eqc10}, we obtain
\begin{equation}\label{eee41}
\sum_{\tau=0}^{t}\|x_\tau\|^{\gamma}=o(\bar{r}_{t})+O(t),\;a.s., t\to \infty.
\end{equation}
Since $\pi_{\theta^{*}}$ is Lipschitz continuous, there exists a constant $L_3>0$ such that $\|\pi_{\theta^{*}}(x_{t})\|\leq L_3\|x_t\|+\|\pi_{\theta^{*}}(0)\|$.
Moreover, by \eqref{assum4}, we also know $\|\pi_{\hat{\theta}_{t-1}}(x_{t})-\pi_{\theta^{*}}(x_{t})\|\leq L_1 \|\theta^{*}-\hat{\theta}_{t-1}\|\|x_t\|+L_2$, which implies $\|\pi_{\hat{\theta}_{t-1}}(x_{t})\|=O(\|x_{t}\|)$ almost surely as $t\to \infty.$
Therefore, by \eqref{eee41}, 
\begin{equation}\nonumber
\begin{aligned}
\sum_{\tau=0}^{t}\|u_\tau\|^{\gamma}=&O\left(\sum_{\tau=0}^{t}\|\pi_{\hat{\theta}_{\tau-1}}(x_{\tau})\|^{\gamma}+\sum_{\tau=0}^{t}\|v_{\tau}\|^{\gamma}\right)\\
=&o(\bar{r}_{t})+O(t),\;a.s.
\end{aligned}
\end{equation}
Since $\phi_t = [x_t^{\top}, u_{t}^{\top}]^{\top}$, the above result gives 
$
\bar{r}_{t}=\sum_{\tau=0}^{t}\|\phi_{\tau}\|^{\gamma}=o(\bar{r}_{t})+O(t),\;a.s.,
$
which implies $\bar{r}_t=O(t),\;\text{a.s.}$
Finally, substituting this into \eqref{eee41}, we conclude that
\begin{equation}\label{eqc14}
\sum_{\tau=0}^{t}\|\phi_\tau\|^{\gamma}=O(t),\;a.s.,
\end{equation}
and Theorem \ref{thm3} holds.
\end{pf}
\section{Proof of Theorem \ref{thm2}}\label{ap4}
For each $t\geq 1$, we define $
\zeta_{t}=[w_t^{\top}, \kappa_{t}^{\top}+v_{t}^{\top}]^{\top},
$ $
\kappa_{t}=\pi_{\hat{\theta}_{t-1}}(x_{t})-\mathbb{E}[\pi_{\hat{\theta}_{t-1}}(x_{t})\mid \mathcal{F}_{t-1}],
$ and $\delta_{t}=\pi_{\hat{\theta}_{t-1}}\left(x_{t}\right)-\pi_{\hat{\theta}_{t-1}}\left(f(\theta^{*\top}\phi_{t-1})\right).$

We first prove the following lemma.
\begin{lemma}\label{lem77}
Under Assumptions \ref{assum_nonlinear}-\ref{assum_stability}, there exists a real-valued random variable $\bar{c}>0$ such that
\begin{equation}\nonumber
\lambda_{\min}\left\{\sum_{\tau=1}^{t}\zeta_{\tau}\zeta_{\tau}^{\top}\right\}\geq \bar{c}t^{1-2b},\;\text{a.s.}, \forall t\geq 0.
\end{equation}
\end{lemma}
\begin{pf}
Since $\pi_{\hat{\theta}_t}(\cdot)$ is $L-$Lipschitz, we have
\begin{equation}
\begin{aligned}
&\|\delta_{t+1}\|=\|\pi_{\hat{\theta}_{t}}\left(x_{t+1}\right)-\pi_{\hat{\theta}_{t}}\left(f(\theta^{*\top}\phi_{t})\right)\|\\
\leq &L\|x_{t+1}-f(\theta^{*\top}\phi_{t})\|\leq L\|w_{t+1}\|.
\end{aligned}
\end{equation}
Hence, by Assumption \ref{assum_noise},
\begin{equation}\label{d1}
\begin{aligned}
\sup_{t\geq 0}\mathbb{E}\left[\|\delta_{t+1}\|^{\gamma}\mid \mathcal{F}_{t}\right]
\leq  \sup_{t\geq 0}L\mathbb{E}\left[\|w_{t+1}\|^{\gamma}\mid \mathcal{F}_{t}\right]<\infty,\text{a.s.}
\end{aligned}
\end{equation}
Since $\pi_{\hat{\theta}_{t}}(f(\theta^{*\top}\phi_{t}))$ is $\mathcal{F}_t-$measurable, we have $\mathbb{E}\left[\delta_{t+1}\mid \mathcal{F}_{t}\right]=\mathbb{E}[\pi_{\hat{\theta}_{t}}(x_{t+1})\mid \mathcal{F}_{t}]-\pi_{\hat{\theta}_{t}}\left(f(\theta^{*\top}\phi_{t})\right).$
Thus, by the definition of $\kappa_{t+1}$, we obtain
\begin{equation}\nonumber
\kappa_{t+1}=\delta_{t+1}-\mathbb{E}\left[\delta_{t+1}\mid \mathcal{F}_{t}\right],\;\forall t\geq 0.
\end{equation}
From this and \eqref{d1}, it follows that
\begin{equation}\label{d3}
\sup_{t\geq 0}\mathbb{E}\left[\|\kappa_{t+1}\|^{\gamma}\mid \mathcal{F}_{t}\right]<\infty.
\end{equation}
Note that $\{\|\kappa_{t+1}\|^{2}-\mathbb{E}[\|\kappa_{t+1}\|^{2}\mid \mathcal{F}_{t}]\}_{t\geq 0}$ is a martingale difference sequence. Then by \eqref{d3} and martingale convergence theorem (cf. Theorem 2.8 in \cite{cg1991}), we have for any $\eta>0,$
$$\sum_{\tau=0}^{t-1}\left(\|\kappa_{\tau+1}\|^{2}-\mathbb{E}\left[\|\kappa_{\tau+1}\|^{2}\mid \mathcal{F}_{\tau}\right]\right)=O(t^{\frac{2}{\gamma}}(\log t)^{\frac{2}{\gamma}+\eta}), \; \text{a.s.}$$
Combining this bound with \eqref{d3} again gives $\sum_{\tau=1}^{t}\|\kappa_{\tau}\|^{2}=O(t),$ a.s. Applying the Cauchy-Schwarz inequality then yields $\sum_{\tau=1}^{t}\|w_{\tau}\kappa_{\tau}^{\top}\|=O(t),$ a.s. 
Hence, there exists a real-valued random variable $c_3>0$ such that
\begin{equation}\label{eqd5}
\sum_{\tau=1}^{t}\left\|w_{\tau}\kappa_{\tau}^{\top}\right\|\leq c_{3}t, \forall t\geq 0.
\end{equation}
Using the martingale convergence theorem again (cf. Theorem 2.8 in  \cite{cg1991}), we also have
$$\sum_{\tau=0}^{t-1}\left(w_{\tau+1}w_{\tau+1}^{\top}-\mathbb{E}\left[w_{\tau+1}w_{\tau+1}^{\top}\mid \mathcal{F}_{\tau}\right]\right)=O(t^{\frac{2}{\gamma}}(\log t)^{\frac{2}{\gamma}+\eta}), \; \text{a.s.,}$$
for any $\eta>0.$ Then, by Assumption \ref{assum_noise},
there exists a real-valued random variable $\bar{\sigma}>0$ such that
\begin{equation}
\begin{aligned}  \sum_{\tau=1}^{t}w_{\tau}w_{\tau}^{\top}=&\sum_{\tau=0}^{t-1}\left(w_{\tau+1}w_{\tau+1}^{\top}-\mathbb{E}\left[w_{\tau+1}w_{\tau+1}^{\top}\mid \mathcal{F}_{\tau}\right]\right)\\
&+\sum_{\tau=0}^{t-1}\mathbb{E}\left[w_{\tau+1}w_{\tau+1}^{\top}\mid \mathcal{F}_{\tau}\right]
\geq \bar{\sigma} t I_{n},\text{a.s.}
\end{aligned}
\end{equation}
Similarly, there exists a real-valued random variable $\bar{\sigma}'>0$ such that $\lambda_{\min}\{\sum_{\tau=1}^{t}\epsilon_{\tau}\epsilon_{\tau}^{\top}\}\geq \bar{\sigma}'t,$ a.s..

For each $z=[z_1, z_2]^{\top}\in \mathbb{R}^{n+m}$ satisfying $\|z\|=1$,
\begin{equation}\label{d6}
\begin{aligned}
z^{\top}&\sum_{\tau=1}^{t}\zeta_{\tau}\zeta_{\tau}^{\top}z
=\sum_{\tau=1}^{t}(z_1^{\top}w_{\tau}+z_2^{\top}\kappa_{\tau})(z_1^{\top}w_{\tau}+z_2^{\top}\kappa_{\tau})^{\top}\\
&+\sum_{\tau=1}^{t}\tau^{-b}(z_{1}^{\top}w_{\tau}+z_{2}^{\top}\kappa_{\tau})\epsilon_{\tau}^{\top}z_2+\tau^{-2b}z_2^{\top}\sum_{\tau=1}^{t}\epsilon_{\tau}\epsilon_{\tau}^{\top}z_2.
\end{aligned}
\end{equation}
Let the filtration $\{\bar{\mathcal{F}}_t\}_t$ be the nondecreasing family of $\sigma-$algebras defined by $\bar{\mathcal{F}}_t=\mathbb{\sigma}(x_0,\{\epsilon_{i}\}_{0\leq i \leq t-1}, \{w_i\}_{0\leq i \leq t})$. Then, it is immediate that  $\kappa_{t}$ is $\bar{\mathcal{F}}_t-$measurable for all $t\geq 0$, and $\mathbb{E}[z_{1}^{\top}w_{t}+z_{2}^{\top}\kappa_{t})\epsilon_{t}^{\top}z_2\mid \bar{\mathcal{F}}_t]=0,$ a.s. Again, applying the martingale convergence theorem, i.e., Theorem 2.8 in  \cite{cg1991}, we obtain
\begin{equation}\label{d7}
\sum_{\tau=1}^{t}(z_{1}^{\top}w_{\tau}+z_{2}^{\top}\kappa_{\tau})\epsilon_{\tau}^{\top}z_2=o\left(\sum_{\tau=1}^{t}\|z_{1}^{\top}w_{\tau}+z_{2}^{\top}\kappa_{\tau}\|^{2}\right),\;a.s.
\end{equation}
Next, we prove Lemma~\ref{lem77} by considering two cases. Case (i): 
If $\|z_{2}\|\leq \bar{\sigma}/4c_{3}$, then $\|z_{1}\|\geq \sqrt{16c_{3}^{2}-\bar{\sigma}^{2}}/4c_3$ and one can verify $4c_3\|z_1\|\|z_2\|t\leq \bar{\sigma}\|z_1\|^{2}t.$ Thus, 
\begin{equation}
\begin{aligned}
&z^{\top}\sum_{\tau=1}^{t}\zeta_{\tau}\zeta_{\tau}^{\top}z\geq  \sum_{\tau=1}^{t}(z_1^{\top}w_{\tau}+z_2^{\top}\kappa_{\tau})(z_1^{\top}w_{\tau}+z_2^{\top}\kappa_{\tau})^{\top}\\
\geq & \bar{\sigma} \|z_1\|^2 t-2c_{3}\|z_1\|\|z_2\|t \geq \frac{\bar{\sigma}\sqrt{16c_{3}^{2}-\bar{\sigma}^{2}}}{8c_3}t.
\end{aligned}
\end{equation}
Case (ii): If $\|z_{2}\|\geq \bar{\sigma}/4c_{3}$, then by \eqref{d6} and \eqref{d7}, we have $
z^{\top}\sum_{\tau=1}^{t}\zeta_{\tau}\zeta_{\tau}^{\top}z\geq \tau^{-2b}z_2^{\top}\sum_{\tau=1}^{t}\epsilon_{\tau}\epsilon_{\tau}^{\top}z_2\geq \frac{\bar{\sigma}\bar{\sigma}'}{4c_{2}}t^{1-2b}.$
From the two cases above, Lemma~\ref{lem77} follows.
\end{pf}

\begin{pf}[Proof of Theorem \ref{thm2}.] We first prove that there exists a real-valued random variable $c_4>0$ such that
\begin{equation}\label{eqd10}
\lambda_{\min}\left\{\sum_{\tau=0}^{t}\phi_{\tau}\phi_{\tau}^{\top}\right\}\geq c_4t^{1-2b}, \text{a.s.}
\end{equation}
From the definition of $\{\zeta_{t}\}$, the sequence $\{\zeta_{t}, \mathcal{F}_t\}$ is a martingale difference sequence, i.e., $\mathbb{E}[\zeta_{t+1}\mid \mathcal{F}_{t}]=0, \forall t\geq 0,$ and it satisfies $\sup_{t}\mathbb{E}[\|\zeta_{t+1}\|^{2}\mid \mathcal{F}_t]<\infty.$ Moreover, $\phi_{t+1}-\zeta_{t+1}$ is $\mathcal{F}_{t}-$measurable for all $t\geq 0$, and using \eqref{eqc14}, \eqref{d3}-\eqref{eqd5}, and Assumption \ref{assum_noise}, 
\begin{equation}\label{eqqd9}
\sum_{\tau=0}^{t-1}\|\phi_{\tau+1}-\zeta_{\tau+1}\|^{2}=O(t), a.s., t\to \infty.
\end{equation}
Then, by the martingale convergence theorem (Theorem 2.8 in \cite{cg1991}), we have
$\sum_{\tau=0}^{t}(\phi_{\tau}-\zeta_{\tau})\zeta_{\tau}^{\top}=O(t^{\frac{1}{2}}(\log t)^{\frac{2}{3}}),$ a.s.
Since $b<\frac{(\gamma-2)^{2}}{4\gamma(\gamma+2)}<\frac{1}{4}$,
\begin{equation}\label{stable}
\sum_{\tau=0}^{t}(\phi_{\tau}-\zeta_{\tau})\zeta_{\tau}^{\top}=o(t^{1-2b}),\;\text{a.s.}
\end{equation}
Note that $\sum_{\tau=0}^{t}\phi_{\tau}\phi_{\tau}^{\top}=\sum_{\tau=0}^{t}(\phi_{\tau}-\zeta_{\tau}+\zeta_{\tau})(\phi_{\tau}-\zeta_{\tau}+\zeta_{\tau})^{\top}
\geq \sum_{\tau=0}^{t}(\phi_{\tau}-\zeta_{\tau})\zeta_{\tau}^{\top}+\sum_{\tau=0}^{t}\zeta_{\tau}(\phi_{\tau}-\zeta_{\tau})^{\top}+\sum_{\tau=0}^{t}\zeta_{\tau}\zeta_{\tau}^{\top}.$ 
Using \eqref{stable} and Lemma~\ref{lem77}, the first two terms are $o(t^{1-2b}),$ a.s., while the last term scales as $t^{1-2b},$ a.s. Hence, \eqref{eqd10} holds.
Moreover, for each $\eta \in (\frac{8b}{\gamma-2}, \frac{2(\gamma-2)}{\gamma(\gamma+2)})$, we have
\begin{equation}\label{eqc13}
\begin{aligned}
&\sum_{\tau=0}^{t}\frac{\phi_{\tau}\phi_{\tau}^{\top}}{1+\|\phi_{\tau}\|^{2}}\geq  \sum_{\tau=0}^{t}\frac{\phi_{\tau}\phi_{\tau}^{\top}}{1+\|\phi_{\tau}\|^{2}}I(1+\|\phi_{\tau}\|^{2}\leq \tau^{\eta})\\
= &\sum_{\tau=0}^{t}\frac{\phi_{\tau}\phi_{\tau}^{\top}}{\tau^{\eta}}-\sum_{\tau=0}^{t}\frac{\phi_{\tau}\phi_{\tau}^{\top}}{\tau^{\eta}}I(1+\|\phi_{\tau}\|^{2}>\tau^{\eta}).
\end{aligned}
\end{equation}
Then, note that
\begin{equation}\label{d13}
\begin{aligned}
&\sum_{\tau=0}^{t}\left\|\frac{\phi_{\tau}\phi_{\tau}^{\top}}{\tau^{\eta}}I(1+\|\phi_{\tau}\|^{2}>\tau^{\eta})\right\|\\
\leq&\sum_{\tau=0}^{t}\left\|\frac{\phi_{\tau}\phi_{\tau}^{\top}}{\tau^{\eta}}\right\|\left(\frac{1+\|\phi_{\tau}\|^{2}}{\tau^{\eta}}\right)^{\frac{\gamma-2}{4}}\\
=&O\left(\sum_{\tau=0}^{t}\frac{\|\phi_{\tau}\|^{2+\frac{\gamma-2}{2}}}{\tau^{\eta(\frac{1}{2}+\frac{\gamma}{4})}}\right),\:\text{a.s.}
\end{aligned}
\end{equation}
By the Cauchy–Schwarz inequality and \eqref{eqc14}, we get
\begin{equation}\label{eqc15}
\begin{aligned}
&\sum_{\tau=0}^{t}\frac{\|\phi_{\tau}\|^{2+\frac{\gamma-2}{2}}}{\tau^{\eta(\frac{1}{2}+\frac{\gamma}{4})}}=\sum_{\tau=0}^{t}\frac{\|\phi_{\tau}\|^{1+\frac{\gamma}{2}}}{\tau^{\eta(\frac{1}{2}+\frac{\gamma}{4})}}\\
\leq &\left[\sum_{\tau=0}^{t} \|\phi_{\tau}\|^{\gamma}\right]^{\frac{\gamma+2}{2\gamma}}\left[\sum_{\tau=0}^{t} \frac{1}{\tau^{\eta\gamma(\gamma+2)/2(\gamma-2)}}\right]^{\frac{\gamma-2}{2\gamma}}\\
\leq& t^{\frac{\gamma+2}{2\gamma}} t^{\frac{\gamma-2}{2\gamma}-\eta(2+\gamma)/4}=t^{1-\eta(2+\gamma)/4}.
\end{aligned}
\end{equation}
Since $\eta> \frac{8b}{\gamma-2}$, we have $1-\eta(2+\gamma)/4<1-2b-\eta.$ Combining this with \eqref{d13} and \eqref{eqc15} gives
\begin{equation}\label{eqc16}
\sum_{\tau=0}^{t}\frac{\phi_{\tau}\phi_{\tau}^{\top}}{\tau^{\eta}}I(1+\|\phi_{\tau}\|^{2}>\tau^{\eta})= o(t^{1-2b-\eta}), \text{a.s.}
\end{equation}
From \eqref{eqc13}, \eqref{eqc16}, and \eqref{eqd10}, we get $t^{1-2b-\eta}=O(\lambda_{t}),$ a.s., where $\lambda_{t}$ is given in \eqref{lam}. With this, Theorem \ref{thm2} follows from property (iii) of Theorem~\ref{thm1}.
\end{pf}

\section{Proof of Theorem \ref{thm4}}\label{ap5}
\vspace{-2mm}
\begin{pf}[Proof of Theorem \ref{thm4}:] From Assumption~\ref{assum_optimality}, there exists a constant $h>0$ such that $\rho_1:=M_0\rho_{0}^{h}\in (0,\frac{1}{2}).$ 
For each fixed $t \ge 0$ and $j \in [0, h-1]$, define
\begin{equation}\nonumber
\begin{aligned}
&\bar{x}_{t+j+1}=f(A^{*}\bar{x}_{t+j}+B^{*}\pi_{\theta^{*}}(\bar{x}_{t+j})+w_{t+j+1}, \bar{x}_{t}=x_{t}.
\end{aligned}
\end{equation}
By Assumption~\ref{assum_optimality}, it holds that
\begin{equation}\label{eq:e3}
\|x_{t+h}^{*}-\bar{x}_{t+h}\|^{2}\leq \rho_1\|x^{*}_{t}-\bar{x}_{t}\|^{2}, \text{a.s}.
\end{equation}
From the $L$-Lipschitz continuity of $\pi_{\theta^{*}}$ and Assumption~\ref{assum_optimality}, it holds that for each $j \in [t, t+h]$,
\begin{align}
&\|u_{j}-\bar{u}_{j}\|= \|\pi_{\hat{\theta}_{j-1}}(x_{j})+v_{j}-\pi_{\theta^{*}}(\bar{x}_{j})\|\nonumber\\ 
\leq& \|\pi_{\theta^{*}}(x_{j})-\pi_{\theta^{*}}(\bar{x}_{j})\|+\|\pi_{\hat{\theta}_{j-1}}(x_{j})-\pi_{\theta^{*}}(x_{j})\|+\|v_{j}\|\nonumber\\
\leq& L\|x_{j}-\bar{x}_{j}\|+ L_1\|\theta^{*}-\hat{\theta}_{j-1}\|\|x_{j}\|+\bar{\epsilon}j^{-b}, \label{eq:D2}
\end{align}
where $\bar{\epsilon}$ is the upper bound of the probing signal $\{\epsilon_t\}$, and $L_1$ is the constant in Assumption \ref{assum_stability}. Applying the bound in \eqref{eq:D2} to the state recursion and using the Lipschitz property of $f(\cdot)$, there exists a constant $\bar{L}>0$ such that
\begin{equation}\nonumber
\begin{aligned}
&\|x_{j+1}-\bar{x}_{j+1}\|\leq \bar{L}\|x_{j}-\bar{x}_{j}\|+\bar{L}\|u_{j}-\bar{u}_{j}\|\\
&\leq(\bar{L}+L\bar{L})\|x_{j}-\bar{x}_{j}\|+\bar{L}L_1\|\theta^{*}-\hat{\theta}_{j-1}\|\|x_{j}\|+\bar{L}\bar{\epsilon}j^{-b}.
\end{aligned}
\end{equation}
From this and the initial equality $\bar{x}_t = x_t$, there exists a constant $M_2 > 0$ such that for each $j \in [t, t+h]$,
\begin{equation}\nonumber
\begin{aligned}
\|x_{t+h}-\bar{x}_{t+h}\|^2\leq M_2\sum_{\tau=t}^{t+h}\|\theta^{*}-\hat{\theta}_{\tau-1}\|^2\|x_{\tau}\|^2+M_2t^{-2b}.
\end{aligned}
\end{equation}
Combining the above bound and \eqref{eq:e3} gives
\begin{equation}\label{eq:e12}
\begin{aligned}
&\|x_{t+h}^{*}-x_{t+h}\|^2\leq 2\|x_{t+h}^{*}-\bar{x}_{t+h}\|^2+2\|\bar{x}_{t+h}-x_{t+h}\|^2\\
\leq & 2\rho_1\|x_{t}^{*}-x_{t}\|^2+2M_2\sum_{\tau=t}^{t+h}\|\theta^{*}-\hat{\theta}_{\tau-1}\|^2\|x_{\tau}\|^2+2M_2t^{-2b}.
\end{aligned}
\end{equation} 
Moreover, by the Cauchy–Schwarz inequality together with Theorems~\ref{thm3} and~\ref{thm2}, we have 
\begin{equation}\label{eq:e13}
\begin{aligned}
&\sum_{k=0}^{\lfloor t/h\rfloor+1}\sum_{\tau=kh}^{kh+h}\|\theta^{*}-\hat{\theta}_{\tau-1}\|^2\|x_{\tau}\|^2\\
=&O((\sum_{\tau=0}^{t+h}\|\theta^{*}-\hat{\theta}_{\tau-1}\|^{2\frac{\gamma}{\gamma-2}})^{1-\frac{2}{\gamma}}(\sum_{\tau=0}^{t+h}\|x_{\tau}\|^{\gamma})^{\frac{2}{\gamma}})\\
=&O([(\log t)^{1+\delta}t^{(2b+\eta)}]^{1-\frac{2}{\gamma}}t^{\frac{2}{\gamma}}),\; \text{a.s.},
\end{aligned}
\end{equation}
where $\lfloor t/h\rfloor$ denotes the integer part of $t/h$. Then, combining \eqref{eq:e12} with \eqref{eq:e13}, we get
\begin{equation}\label{eqd8}
\begin{aligned}
&\sum_{\tau=0}^{t}\|x_{\tau}^{*}-x_{\tau}\|^2\\
=&O(\sum_{k=0}^{\lfloor t/h\rfloor+1}\sum_{\tau=kh}^{kh+h}\|\theta^{*}-\hat{\theta}_{\tau-1}\|^2\|x_{\tau}\|^2)+O(\sum_{\tau=0}^{t}\tau^{-2b})\\
=&O((\log t)^{\frac{(1+\delta)(\gamma-2)}{\gamma}}t^{(2b+\eta-1)(1-\frac{2}{\gamma})+1}+t^{1-2b}), a.s.
\end{aligned}
\end{equation}
Furthermore, note that
\begin{align}
&\|u_{t}^{*}-u_{t}\|^{2}\nonumber\\
=&\|\pi_{\theta^{*}}(x_{t}^{*})-\pi_{\hat{\theta}_{t-1}}(x_{t}^{*})+\pi_{\hat{\theta}_{t-1}}(x_{t}^{*})-\pi_{\hat{\theta}_{t-1}}(x_{t})\|^{2}\nonumber\\
=&O(\|\theta^{*}-\hat{\theta}_{\tau-1}\|^2\|x_{t}^{*}\|^{2})+O(\|x_{t}-x_{t}^{*}\|^{2}),a.s. \nonumber
\end{align}
Then, using \eqref{eq:e13} and \eqref{eqd8}, the following holds almost surely as $t\to \infty,$ 
\begin{align}
&\sum_{\tau=0}^{t}\|u_{\tau}^{*}-u_{\tau}\|^2\nonumber\\
=&O((\log t)^{\frac{(1+\delta)(\gamma-2)}{\gamma}}t^{(2b+\eta-1)(1-\frac{2}{\gamma})+1}+t^{1-2b}).\label{eqd9}
\end{align}
Theorem~\ref{thm4} then follows from \eqref{eg}, \eqref{eqd8}, and \eqref{eqd9}.
\end{pf}
\section{Proof of Lemma \ref{lemsec2}}\label{ap6}
\begin{pf}[Proof of Case (i).] Let $\alpha(\cdot)\equiv\min_{i\in [n]}\alpha_i$ and $\beta(\cdot)\equiv\min_{i\in [n]}\beta_i$. Then $\alpha(\cdot)$ and $\beta(\cdot)$ are positive constant functions, and hence are nonincreasing and nondecreasing, respectively. Moreover, $\inf_{z\in\mathbb{R}^n}\alpha(\|f(z)\|)>0$ and $\sup_{z\in\mathbb{R}^n}\beta(\|f(z)\|)<\infty$. We now verify \eqref{e3}. In fact, for and $c>0$ and all $\|y\|, \|z\|\leq c$,
\begin{equation}\label{f1}
\begin{aligned}
&(y-z)^{\top}(f(y)-f(z))=\sum_{i=1}^{n}(y_{i}-z_i)^{\top}(f_{i}(y_{i})-f_{i}(z_i))\\
&\geq  \sum_{i=1}^{n}\alpha_{i}\|y_i-z_i\|^{2}\geq (\min_{i\in [n]}\alpha_i)\|y-z\|^{2}=\alpha(c)\|y-z\|^{2}.
\end{aligned}
\end{equation}
where $y_i$ and $z_i$ denote the $i$-th components of vectors $y$ and $z$, respectively.  
Similarly, we have $\|f(y)-f(z)\|\leq \beta(c)\|y-z\|.$ Thus, \eqref{e3} holds and Case~(i) follows.

[Proof of Case (ii).] With $\alpha(\cdot)$ and $\beta(\cdot)$ defined in \eqref{e5}, it is clear that  
$\alpha(\cdot)$ is nonincreasing and $\beta(\cdot)$ is nondecreasing, respectively.  Besides, since $\sup_{z\in \mathbb{R}^n}f(z)<\infty$, it follows form \eqref{e4} that $\inf_{z\in \mathbb{R}^n} \alpha(\|f(z)\|)>0$ and $\sup_{z\in \mathbb{R}^n} \beta(\|f(z)\|)<\infty.$ Moreover, for any $c\geq 0$ and $\|y\|, \|z\|\leq c$, by the mean-value theorem, we have
\[
(y_i-z_i)(f_{i}(y_i)-f_{i}(z_i))\geq (\inf_{|x|\leq c} \frac{d}{dx}f_{i}(x))(y_i-z_i)^2.
\]
where $y_i$ and $z_i$ are the $i$-th components of vectors $y$ and $z$, respectively. Then, using an argument analogous to that in \eqref{f1}, we obtain $(y-z)^{\top}(f(y)-f(z))\geq (\inf_{i\in [n], |x|\leq c} \frac{d}{dx}f_{i}(x))\|y-z\|^2=\alpha(c)\|y-z\|^2.$
Similarly, we have $\|f(y)-f(z)\|\leq \beta(c)\|y-z\|.$ Case~(ii) follows.
\end{pf}

 \section{Proof of Lemma \ref{lem6}}\label{ap7}
\begin{pf}[Proof of Lemma \ref{lem6}.] Following the analysis ideas of the classical LS for linear stochastic regression models (see, \cite{lw1982}, \cite{g1995}), we consider the stochastic time-varying Lyapunov function: $$ V_{t}=\operatorname{tr}\left[(\theta^{*}-\hat{\theta}_{t})^{\top}P_{t}^{-1}(\theta^{*}-\hat{\theta}_{t})\right],\; t\geq 0.$$
For each $t\geq 0$, define $$e_{t}=\hat{\theta}_{t}+\mu_{t}^{-1}d_{t}P_{t+1}\phi_{t}[x_{t+1}-f(\hat{\theta}_{t}^{\top}\phi_{t})]^{\top}.$$ 
We first prove the following inequality:
\begin{equation}\label{eqG1}
V_{t+1}\leq \operatorname{tr}\left[(\theta^{*}-e_{t})^{\top}P_{t+1}^{-1}(\theta^{*}-e_{t}) \right].
\end{equation}
To prove \eqref{eqG1}, we consider two cases. \emph{Case (i)}: If $e_{t}\notin \Theta$, then by \eqref{sec3_eq9} and the definition of the projection operator $\Pi_{\Theta}^{t+1}(\cdot)$, we have
$\hat{\theta}_{t+1}=\argmin_{y\in \Theta_{\epsilon}}\|e_{t}-y\|_{P_{t+1}^{-1}}.$ Here, the norm $\|\cdot\|_{P_{t+1}^{-1}}$ is defined by 
$\|x\|_{P_{t+1}^{-1}}=(\operatorname{tr}[x^{\top}P_{t+1}^{-1}x])^{\frac{1}{2}}$ for all $x\in \mathbb{R}^{(m+n)\times n}$ and $t\geq 0$. Since $\Theta_{\epsilon}$ is convex and compact, and $\theta^{*}\in\Theta_{\epsilon}$ by construction, \eqref{eqG1} follows from the non-expansive property of the projection operator (see Lemma~1 in \cite{ZZ2022}).
\emph{Case (ii)}: If instead $e_{t}\in \Theta$, it satisfies 
 $\hat{\theta}_{t+1}=e_{t}$, hence, the inequality \eqref{eqG1} holds with equality. Combining the two cases, \eqref{eqG1} holds for all $t\ge 0$.
 
 From $(\ref{b8})$, we have $P_{t+1}^{-1}=P_{t}^{-1}+\mu_{t}^{-1}d_{t}^{2}\phi_{t}\phi_{t}^{\top}.$
Hence,
\begin{equation}\label{eq_g3}
\begin{aligned}
&(\theta^{*}-\hat{\theta}_{t})^{\top}P_{t+1}^{-1}(\theta^{*}-\hat{\theta}_{t})\\
=&V_{t}+\mu_{t}^{-1}d_{t}^{2}tr[(\theta^{*}-\hat{\theta}_{t})^{\top}\phi_{t}\phi_{t}^{\top}(\theta^{*}-\hat{\theta}_{t})],
\end{aligned}
\end{equation}
and
\begin{equation}\label{28}
	\begin{aligned}
		a_{t}P_{t+1}^{-1}P_{t}\phi_{t}=&a_{t}\left(I+\mu_{t}^{-1}d_{t}^{2}\phi_{t}\phi_{t}^{\top}P_{t}\right)\phi_{t}\\
		=&a_{t}\phi_{t}\left(1+\mu_{t}^{-1}d_{t}^{2}\phi_{t}^{\top}P_{t}\phi_{t}\right)
		=\mu_{t}^{-1}\phi_{t}. 
\end{aligned}
\end{equation}
Substituting \eqref{eq_g3} and \eqref{28} into \eqref{eqG1} yields
\begin{equation}\label{27}
	\begin{aligned}
		V_{t+1}\leq&\operatorname{tr}\left[(\theta^{*}-e_{t})^{\top}P_{t+1}^{-1}(\theta^{*}-e_{t}) \right]\\
		=&V_{t}
		+\mu_{t}^{-1}d_{t}^{2}\operatorname{tr}[(\theta^{*}-\hat{\theta}_{t})^{\top}\phi_{t}\phi_{t}^{\top}(\theta^{*}-\hat{\theta}_{t})]\\
		&-2\mu_{t}^{-1}d_{t}\operatorname{tr}[(\theta^{*}-\hat{\theta}_{t})^{\top}\phi_{t}\psi_{t}^{\top}]\\
        &+\mu_{t}^{-1}a_{t}d_{t}^{2}\phi_{t}^{\top}P_{t}\phi_{t}\operatorname{tr}[\psi_{t}\psi_{t}^{\top}]\\
		&+2\mu_{t}^{-1}a_{t}d_{t}^{2}\phi_{t}^{\top}P_{t}\phi_{t}\operatorname{tr}[\psi_{t}w_{t+1}^{\top}]\\
		&-2\mu_{t}^{-1}d_{t}\operatorname{tr}[(\theta^{*}-\hat{\theta}_{t})^{\top}\phi_{t}w_{t+1}^{\top}]\\
		&+\mu_{t}^{-2}d_{t}^{2}\phi_{t}^{\top}P_{t+1}\phi_{t}\operatorname{tr}[w_{t+1}w_{t+1}^{\top}].
	\end{aligned}
\end{equation}
Moreover, by Assumption \ref{assum_nonlinear}, we have
\begin{equation}\label{term}
tr[(\theta^{*}-\hat{\theta}_t)^{\top}\phi_{t}\psi_{t}^{\top}]\geq 2d_ttr[(\theta^{*}-\hat{\theta}_t)^{\top}\phi_{t}\phi_{t}^{\top}(\theta^{*}-\hat{\theta}_t)].
\end{equation}
Besides, by the definition of $\mu_{t}$ in \eqref{eq6}, we have $d_t\mu_{t}\geq d_t^{2}\overline{g}_t^{2}\phi_{t}^{\top}P_{t}\phi_{t}$, and hence
\begin{equation}\label{g5}
\begin{aligned}
d_{t}\geq \mu_{t}^{-1}d_{t}^{2}\overline{g}_t^{2}\phi_{t}^{\top}P_{t}\phi_{t}\geq a_{t}d_{t}^{2}\overline{g}_t^{2}\phi_{t}^{\top}P_{t}\phi_{t}.
\end{aligned}
\end{equation}
Note that
\begin{equation}\label{eqf7}
tr[\psi_{t}\psi_{t}^{\top}]\leq \overline{g}_{t}^{2}tr[(\theta^{*}-\hat{\theta}_t)^{\top}\phi_{t}\phi_{t}^{\top}(\theta^{*}-\hat{\theta}_t)],
\end{equation}
Then, by \eqref{term}, \eqref{g5}, and \eqref{eqf7}, we have 
\begin{equation}\label{com}
			\begin{aligned}	&\mu_{t}^{-1}d_{t}tr[(\theta^{*}-\hat{\theta}_t)^{\top}\phi_{t}\psi_{t}^{\top}]\geq \mu_{t}^{-1}a_{t}d_{t}^{2}\phi_{t}^{\top}P_{t}\phi_{t}tr[\psi_{t}\psi_{t}^{\top}].
		\end{aligned}
\end{equation}
Now, substituting $(\ref{term})$ and $(\ref{com})$ into $(\ref{27})$, and
summing up both sides of $\left(\ref{27} \right)$ from $\tau=0$ to $t$, we obtain
	\begin{align}
		V_{t}\leq &V_{0}-\sum_{\tau=0}^{t}\mu_{\tau}^{-1}d_{\tau}^{2}tr[(\theta^{*}-\hat{\theta}_\tau)^{\top}\phi_{\tau}\phi_{\tau}^{\top}(\theta^{*}-\hat{\theta}_\tau)]\nonumber\\
		&+2\sum_{\tau=0}^{t}\mu_{\tau}^{-1}a_{\tau}d_{\tau}^{2}\phi_{\tau}^{\top}P_{\tau}\phi_{\tau}tr[\psi_{\tau}w_{\tau+1}^{\top}]\nonumber\\
		&-2\sum_{\tau=0}^{t}\mu_{\tau}^{-1}d_{\tau}tr[(\theta^{*}-\hat{\theta}_{\tau})^{\top}\phi_{\tau}w_{\tau+1}^{\top}]\nonumber\\
		&+\sum_{\tau=0}^{t}\mu_{\tau}^{-2}d_{\tau}^{2}\phi_{\tau}^{\top}P_{\tau+1}\phi_{\tau}tr[w_{\tau+1}w_{\tau+1}^{\top}]. \label{32}
	\end{align}
We now analyze the last three terms on the right-hand side (RHS) of \eqref{32}. 
For the third term on the RHS of \eqref{32}, by the martingale convergence theorem (Theorem 2.8 in \cite{cg1991}) and Assumption~\ref{assum_noise}, we have
\begin{equation}\label{eqf9}
\begin{aligned}
&\sum_{\tau=0}^{t}\mu_{\tau}^{-1}d_{\tau}tr[(\theta^{*}-\hat{\theta}_{\tau})^{\top}\phi_{\tau}w_{\tau+1}^{\top}]
\\
=&o\left(\sum_{\tau=0}^{t}\mu_{\tau}^{-1}d_{\tau}^{2}tr\left[(\theta^{*}-\hat{\theta}_{\tau})^{\top}\phi_{\tau}\phi_{\tau}^{\top}(\theta^{*}-\hat{\theta}_{\tau})\right]\right),\;\;a.s.
\end{aligned}
\end{equation} 
 For the fourth term, note that $a_{\tau}d_{\tau}^{2}\phi_{\tau}^{\top}P_{\tau}\phi_{\tau}\leq 1$ and $tr[\psi_{\tau}\psi_{\tau}^{\top}]\leq \overline{g}_{\tau}^2tr[(\theta^{*}-\hat{\theta}_{\tau})^{\top}\phi_{\tau}\phi_{\tau}^{\top}(\theta^{*}-\hat{\theta}_{\tau})]$. 
Then, again applying Theorem 2.8 in \cite{cg1991} together with \eqref{b9} and \eqref{g5}, we obtain
 \begin{equation}\nonumber
 \begin{aligned}
&\sum_{\tau=0}^{t}\mu_{\tau}^{-1}a_{\tau}d_{\tau}^{2}\phi_{\tau}^{\top}P_{\tau}\phi_{\tau}tr[\psi_{\tau}w_{\tau+1}^{\top}]\\
 =&o\left(\sum_{\tau=0}^{t}(\mu_{\tau}^{-1}a_{\tau}d_{\tau}^{2}\phi_{\tau}^{\top}P_{\tau}\phi_{\tau})^2tr[\psi_{\tau}\psi_{\tau}^{\top}]\right)\\
=&o\left(\sum_{\tau=0}^{t}\mu_{\tau}^{-1}d_{\tau}^2tr[(\theta^{*}-\hat{\theta}_{\tau})^{\top}\phi_{\tau}\phi_{\tau}^{\top}(\theta^{*}-\hat{\theta}_{\tau})]\right),\;\;a.s.
 \end{aligned}
 \end{equation}
For the last term on the RHS of \eqref{32}, applying \eqref{b8} and, for all $k\ge 0$, setting ${\bf w}_k=\mu_{k}^{-1/2}d_k$ and ${\bf A}_k=P_{k+1}^{-1}$ in \cite[Lemma~2(ii)]{lw1982}, we obtain
\begin{equation}\label{eqff11}
\sum_{\tau=0}^{t}\mu_{\tau}^{-1}d_{\tau}^{2}\phi_{\tau}^{\top}P_{\tau+1}\phi_{\tau}
=O\left(\log r_t\right),a.s.
\end{equation}
Let \(D_t= \mu_t^{-1}d_t^{2}\phi_t^{\top}P_{t+1}\phi_t\) for \(t\ge 0\).
By \eqref{eqff11} and the Abel--Dini--Pringsheim theorem (see, e.g., \cite[Theorem 1a]{Hildebrandt}), it follows almost surely that
\begin{align}\nonumber
\sum_{\tau=0}^{t}\frac{\mu_{\tau}^{-1}d_{\tau}^{2}\phi_{\tau}^{\top}P_{\tau+1}\phi_{\tau}}{(1+\log r_{\tau})^{1+\delta}}
=O\left(\sum_{\tau=0}^{t}\frac{D_{\tau}}{(\sum_{k=0}^{\tau} D_{k})^{1+\delta}}\right)
=O\left(1\right).
\end{align}
Since \(\mu_t\ge (1+\log r_t)^{1+\delta}\) for all \(t\ge 0\), we further have
\begin{equation}\label{eqf11}
\sum_{\tau=0}^{t}\mu_{\tau}^{-2}d_{\tau}^{2}\phi_{\tau}^{\top}P_{\tau+1}\phi_{\tau}
=O\left(1\right),\;a.s.
\end{equation}
From \eqref{eqf11} and $\sup_{t\geq 0}\mathbb{E}[\|w_{t+1}\|^{2}\mid \mathcal{F}_t]<\infty$, we have
\begin{align}
\sum_{\tau=0}^{t}\mu_{\tau}^{-2}d_{\tau}^{2}\phi_{\tau}^{\top}P_{\tau+1}\phi_{\tau}\mathbb{E}[\|w_{\tau+1}\|^{2}\mid \mathcal{F}_\tau]
=O\left(1\right).\label{e141}
 \end{align}
Moreover, the sequence $\{\|w_{t+1}\|^{2}-\mathbb{E}[\|w_{t+1}\|^{2}\mid \mathcal{F}_t]\}_{t\ge 0}$ is a martingale difference sequence with respect to $\{\mathcal{F}_t\}$. Thus, by the martingale convergence theorem (Theorem~2.8 in \cite{cg1991}) and \eqref{eqf11}, we obtain 
\begin{equation}\label{eqf15}
\sum_{\tau=0}^{t}\mu_{\tau}^{-2}d_{\tau}^{2}\phi_{\tau}^{\top}P_{\tau+1}\phi_{\tau}(\|w_{\tau+1}\|^{2}-\mathbb{E}[\|w_{\tau+1}\|^{2}\mid \mathcal{F}_{\tau}])=O(1),\text{a.s.}
\end{equation}
Substituting \eqref{eqf9}-\eqref{eqf15} into \eqref{32} yields Lemma~\ref{lem6}. 
\end{pf}   



\end{document}